\documentclass[11pt,letterpaper]{article}

\usepackage[utf8]{inputenc}
\usepackage{amsmath,amssymb,amsthm,mathrsfs,amsfonts}
\usepackage{bm}
\usepackage{color}
\usepackage{enumerate}
\usepackage{physics}
\usepackage{complexity}
\usepackage{soul,xcolor}
\usepackage{textcomp}
\usepackage{verbatim}
\usepackage{tikz}
\usetikzlibrary{quantikz}
\usepackage{caption}
\usepackage{subcaption}

\newtheorem{conjecture}{Conjecture}
\newtheorem{corollary}{Corollary}
\newtheorem{definition}{Definition}
\newtheorem{lemma}{Lemma}
\newtheorem{remark}{Remark}

\newtheorem{theorem}{Theorem}
\newtheorem{problem}{Problem}

\newcommand{\mc}{\mathcal}
\renewcommand{\E}{\mathop{\mathbb E\/}}
\newcommand{\Per}{\mathrm{Per}}

\newclass{\sharpp}{\#P}
\newclass{\cocequalp}{coC_{=}P}
\newfunc{\expp}{exp}
\newcommand{\przero}[1]{{\sf{p}}_{\bf{0}}\left(#1\right)}

\newenvironment{nalign}{
    \begin{equation}
    \begin{aligned}
}{  \end{aligned}
    \end{equation}
    \ignorespacesafterend{}}
    
\usepackage[colorlinks = true]{hyperref}
\usepackage{longtable}

\usepackage{xcolor}
\definecolor{darkred}  {rgb}{0.5,0,0}
\definecolor{darkblue} {rgb}{0,0,0.5}
\definecolor{darkgreen}{rgb}{0,0.5,0}

\hypersetup{
  urlcolor   = blue,         
  linkcolor  = darkblue,     
  citecolor  = darkgreen,    
  filecolor   = darkred       
}

\usepackage[
letterpaper,
top=1in,
bottom=1in,
left=1in,
right=1in]{geometry}

\newtheorem*{theorem*}{Theorem}

\begin{document}
\title{Noise and the frontier of quantum supremacy}
\author{
Adam Bouland\thanks{Simons Institute for the Theory of Computing and Department of EECS, UC Berkeley. \href{mailto:abouland@berkeley.edu}{abouland@berkeley.edu}}
\and 
Bill Fefferman\thanks{Department of Computer Science, University of Chicago. \href{mailto:wjf@uchicago.edu}{wjf@uchicago.edu}}
\and Zeph Landau\thanks{Department of EECS, UC Berkeley. \href{mailto:zeph.landau@gmail.com}{zeph.landau@gmail.com}}
\and Yunchao Liu\thanks{Department of EECS, UC Berkeley. \href{mailto:yunchaoliu@berkeley.edu}{yunchaoliu@berkeley.edu}}
}
\date{}

\maketitle
\thispagestyle{empty}

\begin{abstract}
Noise is the defining feature of the NISQ era, but it remains unclear if noisy quantum devices are capable of quantum speedups. Quantum supremacy experiments have been a major step forward, but gaps remain between the theory behind these experiments and their actual implementations. In this work we initiate the study of the complexity of quantum random circuit sampling experiments with realistic amounts of noise.

Actual quantum supremacy experiments have high levels of uncorrected noise and exponentially decaying fidelities. It is natural to ask if there is any signal of exponential complexity in these highly noisy devices. Surprisingly, we show that it remains hard to compute the output probabilities of noisy random quantum circuits without error correction. More formally, so long as the noise rate of the device is below the error detection threshold, we show it is $\sharpp$-hard to compute the output probabilities of random circuits with a constant rate of noise per gate. This hardness persists even though these probabilities are exponentially close to uniform. Therefore the small deviations away from uniformity are hard to compute, formalizing an important intuition behind Google's supremacy claim. 

Interestingly these hardness results also have implications for the complexity of experiments in a low-noise setting. The issue here is that prior hardness results for computing output probabilities of random circuits are not robust enough to imprecision to connect with the Stockmeyer argument for hardness of sampling from circuits with constant fidelity. We exponentially improve the robustness of prior results to imprecision, both in the cases of Random Circuit Sampling and BosonSampling. In the latter case we bring the proven hardness within a constant factor in the exponent of the robustness required for hardness of sampling for the first time. We then show that our results are in tension with one another -- the high-noise result implies the low-noise result is essentially optimal, even with generalizations of our techniques.
\end{abstract}

\newpage
\clearpage
\pagenumbering{arabic}

\section{Introduction}

The last two years have seen the first claimed demonstrations of ``quantum supremacy'' \cite{arute2019quantum,Zhong2020quantum}-- the 
experimental demonstration of an exponential quantum computational speedup. Quantum supremacy is both a necessary step on the path to building useful quantum computers, and also marks the first experimental challenge to the Extended Church-Turing thesis, a bedrock of theoretical computer science.  
Achieving quantum supremacy requires finding a computational task which is experimentally feasible and has strong complexity-theoretic evidence for classical intractability.
Two prominent supremacy proposals have been BosonSampling \cite{Aaronson2011linear} and Random Circuit Sampling \cite{boixo2018characterizing}, both of which involve sampling from the output distribution of a randomly chosen quantum experiment\footnote{Either a random circuit in the case of RCS, or a random linear optical network in the case of BosonSampling.}.

Despite having advantages in being near-term implementable, it is a priori unclear why these random quantum experiments should be difficult to classically simulate.  While we know of many examples of quantum algorithms that attain exponential speedups over classical computation, they all rely on highly structured circuits (such as quantum Fourier transforms) which are far from typical.  Why then should we expect a generic quantum experiment to realize a large computational advantage?

Aaronson and Arkhipov \cite{Aaronson2011linear} (in the setting of BosonSampling), and  Bouland \textit{et al.} \cite{bouland2019complexity} and Movassagh \cite{movassagh2020quantum} (in the setting of Random Circuit Sampling) gave initial evidence that these random experiments might be difficult to simulate classically\footnote{See also the complementary work of \cite{aaronson2017complexity,aaronson2019classical} for alternative evidence that directly addresses the hardness of scoring well on benchmarking tests such as cross-entropy.}.  
These works essentially show that it is $\#\mathsf{P}$-hard to \emph{exactly compute} the output probabilities of ideal, noiseless, quantum experiments on average\footnote{Please see Section \ref{sec:ourresults_intro} for a discussion of the differences between these works.}.  To do this, they establish a worst-to-average-case reduction for computing the output probabilities of random experiments.  Since it is a well-established fact that computing the output probability of a worst-case quantum experiment is $\#\mathsf{P}$-hard (see e.g., \cite{terhal2002adaptive}), these results prove that computing the output probabilities of a random quantum experiment is $\#\mathsf{P}$-hard as well.

These arguments have two primary limitations which keep them from directly matching the constraints of noisy quantum sampling experiments.  
First, they are sensitive to additive imprecision, which prevents them from addressing the hardness of classically mimicking \emph{low noise} experiments. For example, in the case of random circuit sampling experiments with $m$ gates (where we assume $m$ is greater than the number of qubits, $n$) the prior works are only able to prove hardness of computing an estimate to the output probability that achieves an additive error at most $2^{-O(m^3)}$ \cite{bouland2019complexity, movassagh2020quantum}.
In contrast, standard reductions \cite{Aaronson2011linear,Stockmeyer1985on} require hardness to additive imprecision $O(2^{-n})$ in order to show that the more natural sampling task -- which is what a low-noise\footnote{That is, such a hardness result would imply that no classical algorithm can sample from any distribution within a fixed constant total variation distance of the ideal distribution sampled by the noiseless quantum circuit.  Interestingly, we do not even have hardness of \emph{exact} sampling results for random quantum circuits owing to this same imprecision issue!} random quantum circuit experiment solves ``natively'' -- cannot be reproduced classically. 
A similarly large gap exists between the known ($e^{-O(n^4)}$) vs. conjectured ($O(e^{-n\log n})$) robustness of hardness results for BosonSampling.
Proving a hardness result which is robust even to this small amount of noise remains open for all quantum supremacy proposals.

The second and more fundamental limitation of these arguments is that they say nothing about the complexity of \emph{high noise} experiments.
The prior average-case hardness arguments of \cite{Aaronson2011linear, Bremner2016averagecase, bouland2019complexity,movassagh2020quantum} all explore the complexity of mimicking low-noise experiments performed with a high fidelity that  stays constant as the system size grows. This would require improved control over the noise (e.g. error correction) to achieve asymptotically. 
But in the NISQ era we do not have the quantum resources to implement error correction and so experiments are highly noisy.
For example, Google's recent quantum supremacy experiment estimated that their fidelity was merely $\sim 0.2\%$ (i.e., the experiment was $\sim 99.8\%$ noise) \cite{arute2019quantum} and that their fidelity will further decrease as they scale their system. 
Therefore their quantum supremacy claim hinges on whether or not random circuit sampling is intractable in the high noise regime, in which there is only a small signal of the correct experiment in a sea of noise, and this signal diminishes with system size. 
In some cases noise can render the simulation of certain quantum circuits substantially easier (e.g. \cite{bremner2017achieving,zhou2020limits}).
For this reason, it is critical to  examine if these hardness of computation results are robust to noise models that are relevant to near-term quantum experiments.  
For example, can anything be said about the hardness of computing the output probabilities of such circuits in the high noise regime, without error correction?

In this work we make progress on both these limitations.  First, we substantially improve the robustness of prior hardness results.
For example, in the case of RCS, we show that it is $\#\mathsf{P}$-hard to compute the output probabilities of (noiseless) random circuits to additive imprecision $2^{-O(m\log m)}$.
In the regime of constant depth circuits where $m=O(n)$, this robustness nearly reaches that needed to establish the classical hardness of sampling in the low noise regime, though a small gap remains.

\begin{theorem}[Simplified]\label{thm:simplifiedaveragecasehardness}
For any constant $\eta<1/4$, it is $\sharpp$-hard to compute the output probability of random quantum circuits $C$ up to additive error $2^{-O(m\log m)}$, with probability at least $1-\eta$ over the choice of $C$ with $m$ gates. 
\end{theorem}

We note that Kondo, Mori and Movassagh \cite{kondo21} recently claimed the same robustness as our Theorem \ref{thm:simplifiedaveragecasehardness} and arrived at it independently.

We also give an analogous result for BosonSampling, showing that it is hard to compute the output probabilities of (noiseless) random $n$-photon $m=n^2$ mode linear optical networks to additive imprecision $e^{-6n\log n - O(n)}$ (see Corollary~\ref{cor:permanent}).
This nearly matches the desired robustness for BosonSampling ($O(e^{-n\log n}$)) up to a constant factor in the exponent.
To our knowledge, this is the first time that a quantum supremacy proposal exhibits a proven robustness of hardness which is polynomially related to the conjectured robustness in absolute terms.

Second, we show that these hardness results remain true even in the high noise regime -- i.e. in the presence of a constant rate of noise, without error correction.
For example, consider the case of random circuit sampling, and suppose one fixes a noise model, such as i.i.d. depolarizing noise at rate $\gamma=\Theta(1)$ after each 2-qubit gate.
We show that it is also $\#\mathsf{P}$-hard to compute the output probabilities of the noisy circuits\footnote{Here we mean the output probabilities averaged over the noise, i.e. the corresponding diagonal entries of the density matrix after evolution under the circuit and noise channels.} on average to additive imprecision $2^{-O(m\log m)}$.
This holds so long as the noise rate $\gamma$ is below the error detection threshold.

\begin{theorem}[Simplified]\label{thm:simplifiednoisyhardness}
 The same average-case hardness result as in Theorem~\ref{thm:simplifiedaveragecasehardness} applies to noisy random quantum circuits, assuming the noise is gate-independent and can be error-detected.
 \end{theorem}

This second hardness result is even more interesting when we consider additional structure present in the uncorrected noise setting: the output probabilities of random circuits converge exponentially quickly to the uniform distribution. 
More specifically, it is widely conjectured that the output probabilities of such noisy random quantum circuits are $2^{-O(m)}$ close to uniform~\cite{boixo2018characterizing}.
In fact we rigorously prove this convergence in a toy model in which errors are interspersed with global Haar-random unitaries, though only slower exponential convergences have been shown for random circuit sampling \cite{aharonov1996limitations,gao2018efficient}.
Therefore the output probabilities of the noisy circuits are very close to $2^{-n}$, up to some exponentially suppressed deviations.

This rapid convergence to uniformity has several implications.
First, central to Google's supremacy claim was an assumption that these tiny deviations are in the direction of the ideal distribution\footnote{In particular \cite{arute2019quantum} argued for the hardness of their experiment under a highly simplified noise model of global depolarizing noise, since then the output probabilities are a convex combination of the true distribution with the uniform distribution. Our result generalizes this statement in a more realistic, per-gate noise model, and more generally to any error-detectable noise model.}, and their experimental data was consistent with this claim~\cite{arute2019quantum}. 
Our result supports their intractability claim, in the sense that even these exponentially suppressed deviations from uniformity are hard to compute classically, though it remains open if \emph{sampling} from these highly noisy distributions is classically intractable.

Second, this convergence to uniformity implies that Theorem \ref{thm:simplifiednoisyhardness} is optimal up to the log factor in the exponent!
This is because it cannot possibly be hard to compute the output probabilities of high-noise circuits to additive imprecision $2^{-o(m)}$ -- because the trivial guess of $2^{-n}$ is already within $2^{-O(m)}$ of the noisy probability (on average).
Therefore the robustness of our result to imprecision $2^{-O(m\log m)}$ in the high noise case is essentially tight, as one cannot have a sublinear dependence on $m$ in the exponent. Since our low-noise result is proven with very similar techniques, this means it might be very difficult to improve the robustness of the low-noise result as well.
In this sense our two results are complementary to one another.

\section{Our results and techniques}\label{section:techniques}

Showing both of these results require the development of new techniques to deal with increased amounts of imprecision and noise in the output probabilities. 
For simplicity we will first describe our results in the context of random circuit sampling. We will later show that the same techniques immediately give the analogous results for BosonSampling as well.

\label{sec:ourresults_intro}

\paragraph{Previous work on average-case hardness.}

The core idea of prior average case hardness proofs, such as that of Lipton \cite{Lipton91new}, is to find an algebraic structure in the problem and use it as an error correcting code. 
Specifically, for many $\sharpp$-hard problems, the hard function can be written as a low degree polynomial of the input values.
One can think of this low degree polynomial as a Reed-Solomon code, which allows one to infer the (possibly incorrect) value of a worst-case instance from the (typically correct) values of many average-case instances. 
In order to do so one must find a way of relating the value of a worst case instance to the value of several average-case instances in a way which preserves the low degree polynomial structure.
For example, in the case of permanent of matrices over finite fields, given an arbitrary matrix $A\in\mathbb{F}_q^{n\times n}$, if one picks a random matrix $B$ and lets $A(\theta)=A+\theta B$ for $\theta\in\mathbb{F}_q$, one obtains a family of related instances $\{A(\theta)\}_{\theta\in\mathbb{F}_q}$ which are each individually random, but correlated in a way which preserves this low-degree polynomial structure -- namely, $\Per(A(\theta))$ is a low degree polynomial in $\theta$.
This allows one to infer the value of the worst-case instance $\Per(A(0))$ from many values at $\theta\neq 0$ by polynomial interpolation.

For random quantum circuits similar ideas can be made to work, but creating these correlated families of instances is not straightforward, due to the fact that the algebraic structure of quantum circuits is very different than that of permanents.  
There are currently two different ways to establish this polynomial structure \cite{bouland2019complexity,movassagh2020quantum}.
The original approach of \cite{bouland2019complexity} accomplished this by taking a family of correlated quantum circuits without a low degree polynomial structure, and imposing one by truncating a particular Taylor series.
They used this structure to show the output probabilities of ``truncated'' circuits are hard to compute on average.
While at first glance this average-case result might appear unrelated to the average-case hardness of random circuits -- as the ``input type'' to this problem is different than the random circuit itself, due to the truncations -- surprisingly Bouland \textit{et al.} showed this average-case result is related to the original conjecture, in that robust hardness with respect to this truncated Taylor distribution is equivalent to robust hardness with respect to random circuits (see SI of \cite{bouland2019complexity}). 
Interestingly, while Bouland \textit{et al.}'s result provided the first evidence in support of the approximate average-case supremacy conjecture for random circuits, their work left open whether or not the \emph{true} output probabilities of these circuits are hard to compute on average, as pointed out in the SI of \cite{bouland2019complexity} and discussed in more detail by \cite{movassagh2020quantum,napp2020efficient}.
More recently, Movassagh \cite{movassagh2020quantum} found a clever workaround to avoid the truncation issue.
By simply choosing a different path between the worst and average case instances -- known as the ``Cayley path'' -- he showed it is hard to compute the exact output probabilities of random circuits, using a similar flow of argument to \cite{bouland2019complexity}.
As Movassagh's proof techniques are simpler, and the differences between these approaches could matter in high noise setting (see later discussion in Section \ref{subsec:avgcasehardnessdetails} and Appendix \ref{appendix:truncatedtaylor}), we use Movassagh's interpolation method throughout this work when discussing random circuit sampling results.

More generally, the strategy of establishing average-case hardness via polynomial interpolation is well studied in the context of problems over finite fields.
For example, using the Berlekamp-Welch algorithm \cite{welch1986error} and list decoding techniques, it has been shown that it is hard to compute the permanent of random matrices in $\mathbb{F}_q^{n\times n}$ with constant \cite{gemmell1992highly} or even inverse polynomial \cite{cai1999hardness} probability, so long as the field size $q$ is sufficiently large  (see also \cite{guruswami2006list} for an overview).
However, in the context of quantum supremacy, the relevant polynomials are all over the fields $\mathbb{C}$ or $\mathbb{R}$.
For example, the real-valued output probabilities of quantum circuits are low degree polynomials in the complex gate entries of the input circuit.
This complicates the picture for proving average-case hardness, because first, one needs to consider that some of the values might be incorrect, and second on these correct values, there may be some imprecision in the value, i.e. the correct values may only be within $\pm\delta$ of the true real value\footnote{Indeed, some amount of imprecision is inevitable in any digital model of computation, where real numbers are approximated with poly-sized binary representations to $\delta=2^{-\poly(n)}$ accuracy.}.

\subsection{The bottleneck in prior works}

Existing average-case hardness results over the reals, such as those in \cite{Aaronson2011linear,bouland2019complexity,movassagh2020quantum}, can tolerate only a tiny amount ($\delta=2^{-\poly(n)}$) of such imprecision.
This is for two reasons.
First, many of the sophisticated tools used to show strong average-case hardness over finite fields, such as Berlekamp-Welch and list decoding techniques, do not immediately port over to the real case.
Therefore one is forced to use straightforward Lagrangian interpolation/extrapolation techniques in the average case argument\footnote{We note that if one adopts a model of computation which admits exact representations of real numbers, then the Berlekamp-Welch algorithm does work, as noted in \cite{bouland2019complexity}, but this is not known to work in a digital model of computation where one cannot exactly represent real numbers.}, following the original arguments of Lipton \cite{Lipton91new}. 
Second, it is well known that polynomial interpolation/extrapolation over the reals is exponentially sensitive to noise.
This is essentially because there exist low degree polynomials (like Chebyshev polynomials) which are very close to one another in some interval but diverge rapidly outside of that interval.
In fact very small errors in the value of average-case instances can blow up exponentially when extrapolated to the value of the worst-case instance.

Because of these two limitations, the strongest existing average-case hardness results for computing the output probabilities of random quantum circuits can only tolerate imprecision of order $2^{-O(m^3)}$ in the output probabilities of $n$-qubit $m$-gate quantum circuits. More precisely, Movassagh \cite{movassagh2020quantum}, building on prior arguments of Aaronson and Arkhipov \cite{Aaronson2011linear} and Bouland \textit{et al.} \cite{bouland2019complexity}\footnote{More precisely, \cite{bouland2019complexity} showed (using the techniques of Aaronson and Arkhipov \cite{Aaronson2011linear}) that it is hard to compute  the ``truncated'' circuit probabilities to accuracy $2^{-\text{poly}(m)}$, but did not quantify the polynomial in their argument as it depends on a number of unspecified parameter settings. Movassagh \cite{movassagh2020quantum} improved this to show hardness of computing the true output probabilities, and quantified the robustness to $2^{-O(m^3)}$. }, showed the following:

 \begin{theorem*}[\cite{movassagh2020quantum}, Simplified]\label{thm:oldavgcaseresults}
 It is $\sharpp$-hard to compute the output probability of random quantum circuits $C$ up to additive error $2^{-O(m^{3})}$, with probability $\geq 1-\frac{1}{\poly(n)}$ over the choice of $C$ with $n$ qubits and $m$ gates. 
 \end{theorem*}

The origin of this $O(m^3)$ factor in the exponent is simple to understand.
A well-known result due to Paturi states that if two degree-$d$ polynomials are $\delta$-close to one another in the interval $[0,\Delta]$, then their values at $1$ can differ by at most $\delta\cdot 2^{O(d\Delta^{-1})}$ \cite{Paturi1992on}, which has an exponential dependence in the relative distance $1/\Delta$ and can be saturated by Chebyshev polynomials in many parameter regimes. 
The degree of the relevant polynomial in the context of random quantum circuits is $O(m)$.
The relative distance of interpolation $1/\Delta$ in existing arguments is $O(m^2)$.
Therefore, the Paturi argument shows that polynomial extrapolation over this interval incurs a $2^{O(m^3)}$ blowup in the imprecision when going from the average-case instance to the worst-case instance, which means that a precision of at least $2^{-O(m^3)}$ is required in the average-case output value to infer the worst-case value (which is bounded in $[0,1]$) to any nontrivial level of precision.

Let us be slightly more precise as to why the relative distance of interpolation is $O(m^2)$. 
In these average case arguments, given an arbitrary (worst-case) circuit $C$, one defines a (randomly chosen) family of circuits $C(\theta)$ ($\theta\in[0,1]$) which essentially\footnote{Technically speaking neither the techniques of \cite{bouland2019complexity} nor \cite{movassagh2020quantum} have all of these properties, but the authors found workarounds to this limitation.} has the following properties:
\begin{itemize}
    \item $C(1)=C$, that is, the worst-case circuit corresponds to $\theta=1$.
    \item $C(0)$ is random, and $C(\theta)$ is close to randomly distributed when $\theta$ is small. More specifically, $C(\theta)$ is $O(m\theta)$ close to random in total variation distance for small values of $\theta$.
    \item The probability that the circuit $C(\theta)$ outputs the all-zero outcome $0^n$, denoted as $\przero{C(\theta)}$, is a degree $d=O(m)$ polynomial in $\theta$. 
\end{itemize}
The basic idea of the average-case hardness arguments is to perform standard Lagrangian interpolation.
Consider circuits $C(\theta)$ for many small values of $\theta>0$ that are less than some maximum value $\theta_{\max}=\Delta$ -- say $d+1$ of them.
One then uses the algorithm which works in the average case to compute $\przero{C(\theta)}$ for these circuits.
If the algorithm can correctly compute $\przero{C(\theta)}$ within $\pm \delta$ imprecision on all of them, then these points uniquely define a polynomial $q(\theta)$ by Lagrangian interpolation which is close to the polynomial $\przero{C(\theta)}$ at these evaluated points.
By the Paturi error bound, the polynomial $q$ is $\delta\cdot 2^{O(m\Delta^{-1})}$ close to $\przero{C(\theta)}$ in the interval $[0,1]$\footnote{This is not entirely straightforward, because Paturi's bound on error blowup is for polynomials which are uniformly close in some interval, and these polynomials are only close at particular points, but this turns out not to be an issue as we will describe later. }. 
Therefore one simply uses $q(1)$ as an estimate for the worst case output probability $\przero{C(1)}=\przero{C}$, which is within additive error $\delta\cdot 2^{O(m\Delta^{-1})}$. So long as $\delta \ll 2^{-O(m\Delta^{-1})}$, this will be a good approximation to the desired value.

The reason that existing arguments set $\Delta=O(1/m^2)$  -- leading to a net imprecision tolerance of $\delta=2^{-O(m^3)}$ -- is that Lagrangian interpolation requires \emph{all} of the $O(m)$ evaluated points to be correct (that is, up to imprecision $\delta$). 
So each evaluation point needs to be correct with probability at least $1-O(1/m)$, so that all points are correct with high probability by a union bound\footnote{These points are highly correlated with one another, so unfortunately a union bound is the best tool we know how to use here.}.
However, if one has an algorithm which works on random instances, it will only work with probability $\geq 1-O(m\Delta)$ on each evaluated points, because these points are not actually distributed randomly, but merely $O(m\Delta)$ close to random, as given in the second property defined above.
Therefore one must set $\Delta = O(1/m^2)$ for this argument, to ensure that all of the evaluated points are correct with high probability.

\subsection{Going beyond the imprecision bottleneck}

One can easily see that the value of $\Delta$, as well as the degree of the base polynomial $d=O(m)$, are the controlling factors of the error in this argument.
A natural question is whether or not either of these can be improved.
Increasing the value of $\Delta$ in the interpolation argument would cause some of the evaluated points to be incorrect, which means standard Lagrangian interpolation cannot be used.
In the finite field case, it is well-known how to handle a constant fraction of incorrect points via the Berlekamp-Welch algorithm \cite{welch1986error,gemmell1992highly}.
It is natural to ask if this algorithm, or a variant of it, could be made to work over the Reals.
Critically, in order for such an algorithm to be relevant to our context, we would need the algorithm itself to \emph{not blow up errors too badly} -- or else the gains we make from choosing closer evaluation points could be swamped by the error blowup of the Real version of Berlekamp-Welch.

\subsubsection{First step: robust Berlekamp-Welch over the Reals}

Our first result is to fill in this gap, and show that a variant of Berlekamp-Welch argument still holds over $\mathbb{R}$ or $\mathbb{C}$, without too much blowup in error.
That is, suppose there exists a degree-$d$ polynomial $p(x)$ over the reals, as well as $\poly(d)$ evaluation points $\{x_i\}$ which are equally spaced over the interval $[0,1]$ and claimed values of the polynomial $y_i$ at those points, of which only $1-\eta$ fraction are correct within additive error $\pm\delta$, where $\eta<\frac{1}{4}$ is a constant.
Then our results show that any other polynomial $q(x)$ which also $\delta$-agrees with (possibly different) $1-\eta$ fraction of these evaluation points can diverge from $p(x)$ by at most $\delta\cdot 2^{O(d)}$ at all points in the interval.

Fortunately this error blowup happens to be less than that incurred by our interpolation argument, so is irrelevant to the asymptotics of our result.
This statement is relatively easy to prove if all of the points the polynomials agree on are ``bunched" together on one side of the interval, but extending this to the case where the correct/incorrect points are interspersed with one another requires a very careful argument. 
In particular our argument inductively shows that one can slowly ``bunch" the correct and incorrect points together one by one without loss of generality, without blowing up the error. From there one can then appeal to the bunched case which follows from more standard arguments combining well-known statements about the maximality of Chebyshev polynomial blowups with Markov's inequality for the maximum derivative of polynomials (see Section~\ref{section:robustBW} for details). 

The astute reader will notice that this argument is information theoretic in nature -- that is, if one were to be handed a polynomial which agrees with a large constant fraction of the points, then that polynomial can't diverge ``too badly'' from the true polynomial, but it provides no way of \emph{finding} such a polynomial (which is part of the result of the Berlekamp-Welch algorithm). 
Interestingly, in the context of proving average-case $\sharpp$ hardness, this is all that we need!
That is to say, we only need a form of Berlekamp-Welch error correction with an $\mathsf{NP}$ decoder.
This is because for the purposes of the complexity-theoretic arguments relevant to quantum supremacy, it suffices to prove that computing the output probabilities of random circuits lies outside of the polynomial hierarchy\footnote{This is because the whole point of these arguments is that the output probabilities of classical sampling algorithms can be approximated in the polynomial hierarchy - see Appendix \ref{appendix:truncatedtaylor} for details.}.
Therefore using an $\mathsf{NP}$ decoder in our argument simply means we are showing these probabilities are $\sharpp$ hard to compute under $\mathsf{BPP}^{\mathsf{NP}}$ reductions, which is still outside of $\mathsf{PH}$, unless $\mathsf{PH}$ collapses to a finite level.
Whether an algorithmic version of this result holds over the Reals remains unclear, and is left as an intriguing open problem.

\subsubsection{The next barrier at \texorpdfstring{$2^{-O(m^2)}$}{2\^{}\{-O(m\^{}2)\}} precision}

With this in hand, we have now shown that it is $\#\mathsf{P}$-hard (under $\mathsf{BPP}^\mathsf{NP}$ reductions) to compute the output probabilities of random circuits to additive imprecision\footnote{We note that Movassagh \cite{movassagh2020quantum} independently claimed a similar result, but was missing the robust Berlekamp-Welch theorem necessary for it to be rigorously proven (as also noted in \cite{oszmaniec2020fermion}). Our robust Berlekamp-Welch argument (derived independently from Movassagh's work) completes his proof of this claim, but with the same weakening to hardness under $\mathsf{BPP}^{\mathsf{NP}}$ reductions.} $2^{-O(m^2)}$. 
We did so by decreasing the relative distance $1/\Delta$ we needed to interpolate in our argument from $O(m^2)$ to $O(m)$ in relative terms, by allowing a constant fraction of the interpolation points to be incorrect. 
This proof proceeded by generalizing techniques developed for finite fields to the Reals, but with significant modifications to accommodate the imprecision of real values.

We now wish to push this result closer to achieving hardness to $2^{-O(m)}$ imprecision.
Once again the degree of our polynomial is $O(m)$, so unless we find a different polynomial to work with, the only way forward seems to be to increase the interpolation distance further -- namely to set $\Delta=O(1)$.
A natural question is if more sophisticated techniques developed for finite fields, such as list decoding, might help us to do this, by allowing us to handle a larger fraction of incorrect points, and hence handle a larger interpolation distance.
Unfortunately these sort of ideas will not help us here by more than constant factors.
The issue is that the distribution of circuits $C(\theta)$ are only $O(m\theta)$ close to random in total variation distance.
Therefore, we cannot guarantee we can correctly evaluate \emph{any} points beyond $\Delta=\Theta(1/m)$ with any substantial probability\footnote{For example, if one considers setting $\Delta=\Theta(1)$, the total variation distance between the distribution of $C(\theta)$ and random becomes $1-1/\mathrm{exp}$, and none of the evaluated points will necessarily be correct with non-negligible probability.}.
In fact even proving a variant of list decoding over the Reals would only improve our result by a factor of two in the constant prefactor in the exponent.
As a result, new ideas will be required to tackle this barrier which are not derived from analogues in the finite field case.

\subsubsection{Second step: improved robustness of interpolation over the Reals}\label{sec:rescaling}

Therefore, our second step is to show improved robustness for real-valued polynomial interpolation.
We will first show this using a novel error suppression technique based on variable rescaling which is specific to the Reals.
We will then show that this technique proves that in prior robustness results, the Paturi bound was ``overly pessimistic'' as it was being applied in a parameter regime in which it is not tight!
Therefore one can accurately interpolate degree $d$ polynomials to $2^{-O(d\log \Delta^{-1})}$ additive error, and moreover significantly simplify prior proofs of robustness. 
We provide a pedagogical explanation of this proof below.

To see this, we first note that a variable rescaling trick already shows improved robustness using the standard Paturi bound. 
The basic idea is to consider ``rescaling'' the variables in our problem by replacing our interpolation variable $\theta$ with the variable $x^k$ for some large integer $k\in\mathbb{N}$. 
We then perform our polynomial interpolation algorithm over the variable $x$ instead of the variable $\theta$.
That is, we still consider using the average-case algorithm to compute the output probabilities of many circuits $C(\theta)$ for many values of $\theta<\Delta=O(1/m)$.
But when running the robust Berlekamp-Welch argument, instead of asking the $\NP$ oracle for the degree $O(m)$ polynomial in $\theta$ which $\delta$-approximates the correct values on $1-\eta$ fraction of the evaluating points, and using that to extrapolate our value at $\theta=1$, we instead ask for the same, degree $O(mk)$ polynomial in the variable $x$, and use that polynomial in $x$ to extrapolate our value at $x=1$.

At first glance this sounds strange -- how could simple renaming variables improve the robustness of interpolation? And why would we be using higher degree polynomials, when they increase the error of extrapolation in general?
What is going on here is that we are performing a continuous transformation of the interval $[0,1]$ over which we are interpolating.
That is, consider the map $\Phi:[0,1]\rightarrow[0,1]$ defined as $\Phi(x)=x^k$.
This transformation is bijective, continuous, and preserves the endpoints of the interval. 
Critically, this transformation ``stretches out'' the part of the interval near $0$, and compresses it near $1$.
Therefore, for a fixed value of $\Delta=O(1/m)$, the corresponding value of $x_{\max}$ is substantially increased -- it becomes $1/m^{1/k}$. 
In other words, the interpolation distance is much shorter for this new variable, which decreases the blowup in error from interpolation. 
Of course we pay a price for performing this transformation, as the degree of our polynomial in $x$ is a factor of $k$ larger than the degree of the polynomial in $\theta$.
However this blowup due to an increased polynomial degree is less important than the savings from decreased extrapolation distance,
and in net this results in a substantially reduced error in our extrapolation.

To see this more concretely, recall that the extrapolation error of a degree $d$ polynomial interpolated over distance $\Delta$ is $\delta\cdot  2^{O(d\Delta^{-1})}$ \cite{Paturi1992on}.
Our robust Berlekamp-Welch based algorithm performs interpolation with $d=O(m)$ and $1/\Delta = O(m)$.
After applying our variable rescaling, the degree becomes $O(mk)$, while $\Delta$ becomes $1/m^{1/k}$.
Therefore after this transformation this extrapolation error becomes only $\delta\cdot 2^{k \cdot O(m) \cdot m^{1/k} }$.
By choosing $k=\log m$ this becomes merely a blowup of $2^{O(m\log m)}$. 
Crucially the error $\delta$ in the \emph{output} probability of the quantum circuit is not affected by this rescaling of the \emph{input} variables of the polynomial (namely $\theta$) -- and therefore the value of $\delta$ is the same in the rescaled polynomial.
This therefore allows us to show our main result:
\setcounter{theorem}{0}
 \begin{theorem}[Restated, simplified]\label{thm:simplifiedaveragecasehardness2}
 For any constant $\eta<1/4$, it is $\sharpp$-hard under a $\BPP^{\NP}$ reduction to compute the output probability of random quantum circuits $C$ up to additive error $2^{-O(m\log m)}$, with probability at least $1-\eta$ over the choice of $C$ with $m$ gates. 
 \end{theorem}
 It also shows an analogous hardness result in the case of BosonSampling, as we will describe in Section~\ref{section:avghardness} (see Corollary~\ref{cor:permanent}).

A deeper examination of the above results reveals that this argument proves the Paturi bound itself it not tight in the parameter regime in which it was being applied to quantum supremacy arguments.
The standard justification of the Paturi bound -- that the error of a degree $d$ polynomial interpolated over distance $1/\Delta$ blows up by $ 2^{O(d\Delta^{-1})}$ -- is that Chebyshev polynomials saturate the bound.
This is true in the regime of $\Delta=\Theta(1)$ -- essentially because Chebyshev polynomials have coefficients of order $2^{O(d)}$, and the individual monomials happen to ``cancel out'' in the unit interval, but this carefully orchestrated cancellation quickly falls apart as one leaves the unit interval, and the value of the polynomial approaches $2^{O(d)}$. This behavior is what the Paturi bound tightly captures.

However our variable rescaling trick illustrates that in the regime of asymptotically decreasing $\Delta$, this bound cannot possibly be tight, because nothing in our argument invalidated Chebyshev polynomials\footnote{That is, if $T_d$ is the $d$th Chebyshev polynomial, our errors could uniquely specify $T_d(x^k)$ as our error polynomial, as all we have done is replaced $x$ with $x^k$.}.
In other words, our variable rescaling argument proved that Chebyshev polynomials cannot blow up faster than $2^{O(d\log d)}$ when extrapolated to distance $1/\Delta=O(d)$. 
In retrospect this has a much simpler proof, and it follows because Chebyshev polynomials are after all only polynomials, and not exponential functions!
That is, while Chebyshev polynomials appear to grow exponentially as they ``leap out'' of the small interval, this exponential growth cannot be sustained over longer ranges -- because the function can only grow polynomially quickly in the large distance limit.
To put this more formally, if one has a degree-$d$ polynomial $p(x)=\sum_{i=0}^d c_i x^i$, then  the value of the polynomial at $\Delta^{-1}$ must obey $\left|p(\Delta^{-1})\right| \leq\sum_i |c_i| \Delta^{-d}$. 
So if $\Delta^{-1}=d^c$ for some constant $c>0$, and we know $c_i=2^{O(d)}$, we immediately get that Chebyshev polynomials -- and hence our interpolation error -- is bounded by $d2^{O(d)}\left(d^c\right)^{d} = 2^{O(d\log d)}$.
That is to say, over longer extrapolation distances, Chebyshev polynomials (and in fact any polynomials with bounded coefficients, which is guaranteed in our settings) only have a $2^{O(d\log d)}$ blowup.

Armed with this insight, we provide a significantly simpler proof of robustness of polynomial interpolation to $2^{-O(d\log d)}$ additive imprecision starting from first principles.
This proof is both simpler than our variable rescaling trick, and results in better constants in the exponent.
Moreover, it substantially simplifies prior robustness of interpolation results, as once it is combined with our robust Berlekamp-Welch arguments, it removes the need to invoke heavy-duty techniques of Rahmanov \cite{Rakhmanov2007bounds} in the technical portion of the result and instead only requires applying the much simpler Markov's inequality\footnote{More formally, the previous arguments based on Lagrangian interpolation used a union bound to argue all evaluated points were correct, then invoked Rakhmanov's argument to argue that the polynomial was uniformly bounded in that interval, then invoked Paturi's bound to show that the polynomial remained bounded when extrapolated. Our new proof first argues (using our robust Berlekamp-Welch argument) that the polynomial is uniformly bounded, then applies a simple argument to obtain a bound on the coefficients of the difference polynomial, then directly bounds the error at extrapolation, without the need to invoke Rakhmanov's results.}.
Instead our proof is self-contained and based on fundamental facts about real-valued polynomials.

In the end our final result for random circuit sampling with $m$ 2-qubit gates shows robustness to additive imprecision $\delta=e^{-16m\log m-O(m)}$. For an $n$-photon $m=O(n^c)$-mode BosonSampling experiment, our result shows robustness to additive imprecision $\delta=e^{-(c+4)n\log n-O(n)}$. In contrast the prior leading constants in the exponent of our results for RCS and BosonSampling were 147 and $(c+70)$, respectively, when using the variable rescaling error suppression technique (see Section~\ref{sec:alternativeproof} for a detailed discussion).

\subsection{Extending these arguments to the high-noise regime}

Our second result is to extend these average-case hardness arguments to the high noise regime.
That is, suppose we fix a noise model, such as i.i.d. depolarizing noise at rate $\gamma$.
We show that it is also $\sharpp$-hard (under $\mathsf{BPP}^\mathsf{NP}$ reductions) to estimate the output probabilities of noisy random circuits under this noise model to additive error $2^{-O(m\log m)}$ as well.
This holds so long as the noise rate $\gamma$ is below the error \emph{detection} (not correction!) threshold. 
More generally, our result holds for any stochastic noise model which admits error detection, and which is gate-independent.

There are two key ideas which enable this strengthening of our result.
The first is to note that our worst to average-case reduction techniques apply essentially immediately to the noisy case!
This is because for stochastic noise models like i.i.d. depolarizing noise, the output probability of the noisy circuit is a convex combination of the output probabilities of many different noiseless circuits.
To be more precise, recall that depolarizing noise with Pauli noise rate $\gamma$ after a two-qubit gate is equivalent to applying one of the $15$ different nontrivial Pauli operators (such as $I\otimes X$ or $Y\otimes Z$) with probability $\gamma/15$ after that gate, or else applying the identity operator with probability $1-\gamma$.
Therefore, applying depolarizing noise after every gate of an $m$-gate circuit is equivalent to taking the original circuit $C$, and inserting into the circuit one of the $16^m$ different ``noise patterns'' $\xi$ that could have occurred in the circuit, with appropriately weighted probabilities -- call these circuits $C_\xi$. 
The output probability of the noisy circuit is the appropriately weighted convex combination of the output probabilities of $C_\xi$. 
In other words, if $\przero{C,\mc N}$ is the probability that the noisy circuit outputs $0^n$ under a noise model $\mc N$, then we have that
\begin{equation}
    \przero{C,\mc N} = \E_{\xi \sim \mc T_{\mc N}} \left[\przero{C_\xi}\right]
\end{equation}
where $\mc T_{\mc N}$ is the distribution over error patterns arising from the particular noise model $\mc N$.
This is true in general for stochastic noise models.

The key point to note is that our average-case hardness arguments are invariant under taking convex combinations of circuits.
This is because the key fact we use in our reduction is that, for some appropriately defined family of circuits $C(\theta)$, $\przero{C(\theta)}$ is a low degree polynomial in $\theta$.
The set of low degree polynomials is convex.
Therefore, if one simply defines $C_\xi(\theta)$ to be the same circuits $C(\theta)$ with the noise pattern $\xi$ inserted into the circuit, then we have that
\begin{equation}
    \przero{C(\theta),\mc N} = \E_{\xi \sim \mc T_{\mc N}}\left[\przero{C_\xi(\theta)}\right]
\end{equation}
is also a low degree polynomial in $\theta$.
As a result the exact same worst-to-average case reduction arguments work in the noisy setting (see Section~\ref{sec:noisyworstcase} and \ref{sec:noisyaveragecasehardness} for more details).
Interestingly, for this to work there is no need to know the noise rate $\gamma$ or noise model to apply this argument -- because the low degree polynomial structure is preserved by \emph{any} noise model that is stochastic and gate-independent.

We have therefore established that average-case noisy output probabilities are as hard to estimate as worst-case ones.
The second step required is to establish that these noisy worst-case probabilities are indeed $\sharpp$-hard to compute\footnote{More formally, we prove $\cocequalp$ hardness, which is hard enough to suffice for quantum supremacy arguments as $\cocequalp$ is outside of the polynomial hierarchy unless $\PH$ collapses.}.
To do so we follow the observations first developed by Fujii \cite{fujii2016noise}.
In particular, Fujii observed that no polynomial time algorithm can exactly\footnote{That is, up to constant multiplicative error on each of the exponentially small output probabilities.} sample from the output distribution of noisy quantum circuits in the worst case unless $\PH$ collapses, so long as the noise rate $\gamma$ is below the error detection threshold.
The basic idea is that quantum error detection is highly efficient.
With polynomial overhead, one can create a quantum error detecting code which is exponentially sensitive to errors. 
That is, if no error is detected, then one is certain with probability $1-1/\mathrm{exp}$ that no error indeed occurred.
Using this, Fujii showed that one can post-select the ``no error detected" event so that the post-selected distribution is exponentially close to the ideal output distribution. 
By standard post-selection arguments along the lines of \cite{Aaronson2005quantum,Bremner2011classical} this suffices to show worst-case hardness.
A modification of this argument can be used to show that computing these output probabilities is essentially $\sharpp$-hard.
This establishes the worst-case hardness of noisy circuits, which allows us to show the following theorem:

\begin{theorem}[Restated, simplified]\label{thm:simplifiednoisyhardness2}
 The same average-case hardness result as in Theorem~\ref{thm:simplifiedaveragecasehardness2} applies to noisy random quantum circuits, assuming the noise is gate-independent and can be error-detected.
 \end{theorem}

Note that there is one restriction which applies to this theorem which is not present in the statement of Theorem~\ref{thm:simplifiedaveragecasehardness}.
Namely, this theorem assumes that the family of random circuits considered is deep enough to encode a quantum error detecting code which works against the chosen noise model and noise rate.
This is because the worst-case circuits used in the argument of Theorem~\ref{thm:simplifiednoisyhardness} are error detecting circuits, while constant depth circuits suffice for Theorem~\ref{thm:simplifiedaveragecasehardness}.
Therefore our result essentially says that the output probability of noisy random circuits are as hard to compute as the \emph{largest error detected circuit one could have fit} into that particular random circuit architecture.

\section{Discussion}
\label{sec:discussion}

In this work we have sought to resolve a fundamental disconnect between the theoretical basis of quantum supremacy -- based on very \emph{low noise} random quantum circuits \cite{terhal2002adaptive,Bremner2011classical,Aaronson2011linear,Bremner2016averagecase,bouland2019complexity,movassagh2020quantum} -- and the highly noisy experimental implementations performed so far \cite{arute2019quantum,Zhong2020quantum}.
Namely, the initial theoretical attempts at proving ``quantum supremacy'' for sampling from random quantum circuits sought to do so assuming one can maintain a large constant fidelity as the system size grows -- which would require error correction to achieve scalably.
Even establishing this is classically intractable remains a major theoretical challenge, due to the sensitivity of the average-case hardness of computing results to small amounts of imprecision -- a limitation that we made substantial progress on in our first result (Theorem \ref{thm:simplifiedaveragecasehardness}) but did not resolve.
However the initial experimental attempts at demonstrating supremacy showed a fidelity that \emph{decreases} exponentially with system size, because quantum error correction is not being performed. 
As these sampling experiments become exponentially close to uniform asymptotically \cite{aharonov1996limitations,gao2018efficient}, one cannot argue that these experiments are hard to reproduce in an approximate sense.

Consequently, if there is any hope of finding a hard quantum signal present in these near-term experiments, it must be present in the exponentially diminishing deviations from the uniform distribution.
A natural question is whether or not these small deviations from uniformity are hard to reproduce classically.
Our second main result (Theorem \ref{thm:simplifiednoisyhardness}) provides evidence in this direction, by extending the prior hardness of exact computation results for random quantum circuits to the noisy setting.  That is, we prove it is hard to exactly \emph{compute} the output probability of a noisy random quantum circuit, to finer precision than the rate of convergence to uniform.  This can be interpreted as proving the existence of a quantum signal that survives the noise, thereby formalizing an intuition that underlies the optimism that the recent Google experiment is solving a hard problem.

A natural hope is that one might be able to use these results to show that one cannot \emph{exactly} reproduce the results of the Google experiment in a \emph{sampling} sense.
That is, can one argue that no classical algorithm could sample from the output distribution of noisy random quantum circuits nearly exactly?
Unfortunately we know we cannot do this using current techniques \cite{Aaronson2011linear}.
The reason is that these techniques reduce hardness of exact (or approximate) sampling to the hardness of \emph{approximately} computing output probabilities using Stockmeyer's algorithm \cite{Stockmeyer1985on}, and the amount of imprecision required to be tolerated ($O(2^{-n})$) is intrinsic in the argument, even in the exact sampling case\footnote{Technically speaking, the imprecision incurred by Stockmeyer's algorithm is multiplicative, and so depends on the output probability being estimated, which for random circuits is $O(2^{-n})$ with very high probability due to known ``anti-concentration" results \cite{Harrow2009random,brandao2016local,harrow2018approximate,barak2020spoofing,dalzell2020random}.}. 
As previously discussed, the output probabilities of noisy random circuits cannot be hard to compute to this degree of imprecision, due to the fact the noisy distributions are exponentially close to uniform.
While \cite{aaronson2017complexity,aaronson2019classical} took steps to try to address the simulation of such noisy experiments via a different angle, their approach assumes the fidelity remains a small constant as the system size scales, and requires making non-standard assumptions.
Therefore, we believe fundamentally new ideas will be required -- beyond the Stockmeyer approach -- if one wishes to address the hardness of exactly simulating highly noisy circuits with decreasing fidelities.

Our main results also offer an interesting new message about the low-noise setting: any potential means for proving the hardness of sampling in the \emph{low noise} setting cannot be extended to hold in the \emph{high noise} setting.
Our barrier holds in the setting of random circuits with superconstant depth\footnote{More formally, the conjectured rate of convergence of noisy circuits to the uniform distribution is $O(e^{-dn\gamma})$ where $n$ is the number of qubits and $d$ is the depth of the circuit, and $\gamma$ is the depolarizing noise rate per qubit. The robustness required for proving hardness of sampling in the low noise setting is $O(2^{-n})$. Therefore, for any circuit with $d>1/\gamma_0$, where $\gamma_0$ is the error detection threshold, any noise-agnostic proof technique cannot prove sufficient robustness to prove hardness of sampling in the low noise setting. Note the depth $d$ must be sufficiently large to perform error detection as well.}.
This offers a significant barrier for improving the status quo of quantum supremacy results: the only chance of resolving the hardness of sampling problem (using the standard and only approach we have -- via a reduction to approximate computing using Stockmeyer's theorem) is by inventing new techniques that are explicitly \emph{not resilient} to uncorrected noise!  This immediately rules out any ``convex method'' -- i.e. any method where the underlying error correcting code structure is invariant under convex combinations -- including polynomial interpolation, which we prove in Theorem \ref{thm:simplifiednoisyhardness} is noise resilient.

Our results are complementary, but distinct from past barrier results, starting with that of Aaronson and Arkhipov \cite{Aaronson2011linear}, who first realized that polynomial interpolation cannot close this gap if the ensemble of circuits anticoncentrates (See Appendix~\ref{sec:barrier} for more details).
While there are subtle technical differences between our barriers\footnote{In particular our barrier does not require anticoncentration, but instead requires a distribution over random circuits which converges to uniformity under depolarizing noise. Both properties are related to the scrambling nature of random circuits, but it is not \emph{a priori} clear they are identical.}, they arise for similar reasons -- namely the set of low degree polynomials is convex.
Later, Aaronson and Chen \cite{aaronson2017complexity} proved an oracle separation relative to which approximate quantum sampling is classically easy, but the $\PH$ remains infinite.  While this rules out any relativizing technique for resolving the hardness of approximate sampling, their result is very sensitive to the distinction between approximate and exact sampling -- indeed, worst-case hardness of exact sampling results do exist and relativize \cite{Bremner2011classical}!  Therefore it is unclear how to extend their barrier result to our setting in which we are interested in the hardness of exact sampling from the output distribution of high-noise and/or ideal random quantum circuits.

Most recently, a barrier result for hardness of sampling was shown by Napp \textit{et al.} \cite{napp2020efficient} who give an approximation algorithm for computing the output probabilities of very shallow (i.e., fixed constant) depth 2D random quantum circuits.  This is interesting, because one can use the methods of Bouland \textit{et al.} and Movassagh to prove the average case hardness of \emph{exactly} computing the output probabilities of the same ensemble of circuits.  The implication is that any technique that establishes hardness of computing output probabilities of random quantum circuits to within additive error $O(2^{-n})$ must not work for very shallow depth quantum circuits.  This is a limitation, since current techniques that use polynomial interpolation are not sensitive to depth.  This result is nicely complementary to our barrier, and in fact we can use it to conclude that any technique that extends the hardness of computation results to hardness of sampling results must be sensitive to both depth \emph{and} noise.

Interestingly the noise barrier identified in this paper, and the depth barrier identified by Napp \textit{et al.} \cite{napp2020efficient}, rely on specific properties of random quantum circuits. We leave it as an open problem if analogous barriers apply to the BosonSampling case\footnote{We note this might be possible using the techniques of e.g. \cite{garcia2019simulating,qi2020regimes,oh2021classical} for simulating noisy BosonSampling experiments, though there are many details to be worked out.}. Of course it is already known that polynomial interpolation will not prove hardness of sampling for BosonSampling \cite{Aaronson2011linear} -- in the case of $m=n^2$ modes this barrier shows polynomial interpolation cannot prove robust hardness beyond $O(e^{-3n\log n})$ additive error. This leaves open the possibility of improving the exponent in our robustness result for BosonSampling from $6$ to $3$ using current techniques, but not beyond. 
However, the fact that our results for BosonSampling are closer to achieving the desired robustness -- due to structural differences in the algebraic properties of BosonSampling vs. RCS --  indicates that BosonSampling could have certain advantages for proving ``quantum supremacy'' in the low-noise regime going forward.

\section{Average-case hardness}\label{section:avghardness}
In this section we develop the average-case hardness of computing the output probability of random quantum circuits. We first establish improved results in the noiseless setting, and then show that our techniques can be generalized to the noisy case. An important ingredient of our proof is a robust polynomial interpolation technique, which is independently developed in Section~\ref{section:robustBW}.

\subsection{Worst-to-average-case reduction}
\label{subsec:avgcasehardnessdetails}

Our first result is to establish a worst-to-average-case reduction for computing the output probability of random quantum circuits in the noiseless setting. That is, if there exists an algorithm that can approximate the output probability of \emph{most} random quantum circuits over an architecture, then there exists an algorithm that can approximate the output probability of \emph{every} quantum circuit over the same architecture. Our average-case hardness result tolerates a constant failure probability (i.e. it is hard to compute the output probability for a large constant fraction of random circuits), and is robust up to a small additive error $2^{-O(m\log m)}$ for random circuits with $m$ gates. Our result is therefore an improvement over previous results~\cite{bouland2019complexity,movassagh2020quantum} in both aspects.

The main idea of the worst-to-average-case reduction works as follows. Consider a circuit architecture over which computing the output probability up to an exponentially small additive error is $\sharpp$-hard in the worst case. Then, to prove average-case hardness, it suffices to construct a reduction that belongs to some finite level of the Polynomial Hierarchy. That is, given an arbitrary circuit $C_0$, our goal is to compute $\przero{C_0}$ using a procedure in the Polynomial Hierarchy, with access to an oracle $\mc O$ that can compute $\przero{C}$ for a large fraction of random circuits $C$. For a circuit architecture $\mc A$, let $\mc H_{\mc A}$ be the distribution over circuits such that each gate is independently drawn from the Haar measure. Let $\{G_i\}_{i=1,\dots,m}$ be the quantum gates in $C_0$. We create a new circuit $C_1$ by applying a ``one-time pad" to $C_0$, where each gate $G_i$ is replaced by
\begin{equation}
    G_i\to H_i G_i,
\end{equation}
where $\{H_i\}$ is independently drawn from the Haar measure over the unitary group and has the same dimension as $G_i$. By the invariance of Haar measure, $C_1$ is distributed the same as $\mc H_{\mc A}$. Therefore $\{H_i\}$ can be understood as the ``random seed" used for the one-time pad.

By definition, $\mc O$ can compute an accurate approximation of $\przero{C_1}$ with high probability. However, this number alone does not contain any information on our desired quantity $\przero{C_0}$. The main insight that allows us to correlate average-case solutions to the worst-case quantity, which was originally developed to show the average-case hardness for permanents~\cite{Lipton91new,Gemmell1991self}, is to create many random instances $\{C_i\}$ by the following procedure: first sample a ``random seed" $\{H_i\}$, then apply small and different perturbations to the random seed, and then apply each perturbed random seed to $C_0$. As a result, while each random instance is marginally distributed approximately according to $\mc H_A$, they are close to each other and are correlated in a way that reveals the worst-case quantity $\przero{C_0}$.

More specifically, suppose there is a way of perturbing the circuit $C_1$ into a new circuit $C(\theta)$ ($\theta\in[0,1]$), such that $C(\theta)\approx C_1$ when $\theta\ll 1$, and $C(1)=C_0$. Moreover, suppose $\przero{C(\theta)}$ for different values of $\theta$ are correlated in a way such that $\przero{C(1)}$ can be inferred from $\przero{C(\theta)}$ for small values of $\theta$. Then, to compute $\przero{C_0}$, it suffices to query $\mc O$ with $C(\theta_i)$ for many small $\theta_i$s.

One way to develop such a procedure is by perturbing the random seeds $\{H_i\}$ used in the construction of $C_1$, which we define as follows.

\begin{definition}[$\theta$-perturbed random circuit distribution]\label{def:thetaperturbeddistribution}
Suppose there is a transform on unitary matrices $H\mapsto H(\theta)$ parameterized by a real number $\theta\in[0,1]$, that satisfies $H(0)=H$ is unchanged, and $H(1)=I$ is the identity matrix. For any circuit $C_0$ with gates $\{G_i\}_{i=1,\dots,m}$ and a random seed $\{H_i\}_{i=1,\dots,m}$, the circuit $C(\theta)$ is defined by replacing $G_i$ with 
\begin{equation}
    G_i\mapsto H_i(\theta) G_i.
\end{equation}
Denote the distribution induced by $C(\theta)$ as $\mc H_{\mc A,\theta}$. Note that $\mc H_{\mc A,0}=\mc H_{\mc A}$ and $\mc H_{\mc A,1}=\mathbf{1}_{C_0}$.
\end{definition}

As discussed above, in order for such a perturbation to be useful for the worst-to-average-case reduction, we require that the following (informal) properties are satisfied.

\vspace{2mm}\noindent\textbf{Properties of the perturbed circuit distribution (informal).}
\begin{enumerate}
    \item When $\theta$ is small, $\mc H_{\mc A,\theta}$ is close to $\mc H_{\mc A}$ in total variation distance. Therefore, when $\mc O$ is given an input $C(\theta)\sim \mc H_{\mc A,\theta}$ from the perturbed distribution, it is guaranteed to return a correct approximation of $\przero{C(\theta)}$ with high success probability.
    \item An approximation of $\przero{C(1)}$ can be inferred from approximations of $\przero{C(\theta)}$ for small values of $\theta$ with a procedure in the Polynomial Hierarchy.
\end{enumerate}

Finding such a good perturbation method is non-trivial. Two proposals were developed in previous work in the context of noiseless circuits. In~\cite{bouland2019complexity}, they considered the truncated Taylor series of $H e^{-\theta\log H}$, which is given by
\begin{equation}\label{eq:taylortruncation}
    H(\theta)=H\cdot\left(\sum_{k=0}^K\frac{(-\theta\log H)^k}{k!}\right),
\end{equation}
and showed that $\przero{C(\theta)}$ is a degree $O(mK)$ polynomial in $\theta$. Property 2 is then satisfied by using polynomial interpolation techniques. In this approach, $H(\theta)$ as defined by Eq.~\eqref{eq:taylortruncation} is not unitary, and the resulting average-case hardness is for a circuit family that has a small deviation from $\mc H_{\mc A}$ given by the truncation error of the Taylor series. While this truncation error does not affect the applicability of this approach to address hardness of approximate sampling, it becomes evident when applying to noisy circuits, as the requirement for accuracy is much higher (see appendix~\ref{appendix:truncatedtaylor} for a detailed discussion).

Later, Movassagh~\cite{movassagh2020quantum} developed a new perturbation method, called Cayley transform, that is able to stay within the unitary path.

\begin{definition}[Cayley transform~\cite{movassagh2020quantum}]\label{def:cayleytransform}
The Cayley transform of a unitary matrix $H$ parameterized by $\alpha\in[0,1]$ is a unitary matrix defined as
\begin{equation}
    H(\alpha):=\frac{\alpha I+(2-\alpha)H}{(2-\alpha)I+\alpha H},
\end{equation}
where $A/B$ means $A\cdot B^{-1}$. Consider the diagonalization of a unitary matrix $H=\sum_j e^{i\varphi_j}\ketbra{\psi_j}$, the following is an equivalent form of the Cayley transform,
\begin{equation}
    H(\alpha)=\sum_j\frac{1+i(1-\alpha)\tan\frac{\varphi_j}{2}}{1-i(1-\alpha)\tan\frac{\varphi_j}{2}}\ketbra{\psi_j}.
\end{equation}
Note that $H(0)=H$ and $H(1)=I$.
\end{definition}

Note that we are using slightly different notations from~\cite{movassagh2020quantum}. A self-consistent presentation of the properties of the Cayley transform is given in Appendix~\ref{appendix:cayleytransform}. 

While not having an evident polynomial structure as in the construction of~\cite{bouland2019complexity}, in~\cite{movassagh2020quantum} they applied Cayley transform to the gates $\{H_i\}$ parameterized by $\alpha=\theta$, and showed that the resulting $\przero{C(\theta)}$ is a degree $(O(m),O(m))$ rational function in $\theta$. A similar polynomial interpolation technique can be applied to infer $\przero{C(1)}$ from $\przero{C(\theta)}$ with small values of $\theta$.

In particular \cite{movassagh2020quantum}, building on~\cite{bouland2019complexity}, obtained the following average-case hardness results (also see Appendix~\ref{appendix:truncatedtaylor}).

\begin{theorem}[\cite{bouland2019complexity,movassagh2020quantum}]
Let $\mc A$ be a circuit architecture so that computing $\przero{C}$ to within additive error $2^{-O(m)}$ is $\sharpp$-hard in the worst case. The following results hold:
\begin{enumerate}
    \item It is $\sharpp$-hard to compute $\przero{C}$ exactly on input $C\sim \mc H_{\mc A}$, with success probability at least $1-\eta$ over the choice of $C$, for any constant $\eta<\frac{1}{4}$.
    \item It is $\sharpp$-hard to compute $\przero{C}$ up to an additive imprecision $2^{-O(m^3)}$ on input $C\sim \mc H_{\mc A}$, with success probability at least $1-O(1)/m$ over the choice of $C$.
\end{enumerate}
\end{theorem}

We improve these results by proving an average-case hardness that tolerates a larger additive imprecision $2^{-O(m\log m)}$, while only requiring a constant success probability over the choice of random circuits. These improvements follows from two new techniques: the first is a robust polynomial interpolation argument which tolerates a constant failure probability, and the second is an improved error bound for long distance polynomial extrapolation. Both results are developed in Section~\ref{section:robustBW} and are used as black boxes here.

To establish our reduction, we first consider the $\theta$-perturbed random circuit distribution instantiated by the Cayley transform.

\begin{definition}[$\theta$-Cayley perturbed random circuit distribution]\label{def:thetacayleyperturbation}
For any circuit $C_0$ with gates $\{G_i\}_{i=1,\dots,m}$ and a random seed $\{H_i\}_{i=1,\dots,m}$, the circuit $C(\theta)$ is defined by replacing $G_i$ with 
\begin{equation}
    G_i\mapsto H_i(\theta) G_i,
\end{equation}
where $H_i(\theta)$ denotes the Cayley transform of $H_i$ parameterized by $\theta$. Denote the distribution induced by $C(\theta)$ as $\mc H_{\mc A,\theta}$. Note that $\mc H_{\mc A,0}=\mc H_{\mc A}$ and $\mc H_{\mc A,1}=\mathbf{1}_{C_0}$.
\end{definition}
As previously shown by~\cite{movassagh2020quantum}, the output probability $\przero{C(\theta)}$ is a low-degree rational function in $\theta$.

\begin{lemma}[\cite{movassagh2020quantum}]\label{lemma:lowdegreerational}
For any circuit $C_0$, let $C(\theta)$ be a circuit from the $\theta$-Cayley perturbed random circuit distribution as in Definition~\ref{def:thetacayleyperturbation}. Then $\przero{C(\theta)}$ is a degree $(O(m),O(m))$ rational function in $\theta$.
\end{lemma}
\begin{proof}
Here we give a self-consistent proof as the details are useful for our later developments.

First, notice that the output amplitude can be written as the Feynman path integral,
\begin{equation}\label{eq:feynmanpathsum}
    \expval{C(\theta)}{0^n}=\sum_{y_1,\dots,y_{m-1}\in\{0,1\}^n}\prod_{j=1}^m \mel{y_j}{\left(H_j(\theta)G_j\otimes I_{else}\right)}{y_{j-1}},
\end{equation}
where we have $y_0=y_m=0^n$, $H_j(\theta)G_j$ is a local gate and $I_{else}$ denotes identity on all the other qubits. Let the diagonalization of $H_j$ be $H_j=\sum_{l}e^{i\varphi_{jl}}\ketbra{\psi_{jl}}$. Then, consider an individual term in the above sum,
\begin{nalign}\label{eq:feynmanpathrational}
    \prod_{j=1}^m \mel{y_j}{H_j(\theta)G_j}{y_{j-1}}&=\prod_{j=1}^m \bra{y_j}\sum_l\frac{\left(1+i(1-\theta)\tan\frac{\varphi_{jl}}{2}\right)\ketbra{\psi_{jl}}}{1-i(1-\theta)\tan\frac{\varphi_{jl}}{2}}G_j \ket{y_{j-1}}\\
    &=\prod_{j=1}^m \bra{y_j}\frac{\sum_l\left(1+i(1-\theta)\tan\frac{\varphi_{jl}}{2}\right)\ketbra{\psi_{jl}}\prod_{t\neq l}\left(1-i(1-\theta)\tan\frac{\varphi_{jt}}{2}\right)}{\prod_l\left(1-i(1-\theta)\tan\frac{\varphi_{jl}}{2}\right)}G_j \ket{y_{j-1}}.
\end{nalign}
Let
\begin{nalign}
    Q_0(\theta)&:=\prod_{j=1}^m \prod_l\left(1-i(1-\theta)\tan\frac{\varphi_{jl}}{2}\right),\\
    Q(\theta)&:=|Q_0(\theta)|^2.
\end{nalign}
Then Eq.~\eqref{eq:feynmanpathrational} is a degree $(O(m),O(m))$ rational function in $\theta$ with $Q_0(\theta)$ being the denominator. Furthermore, notice that $Q_0(\theta)$ does not depend on the Feynman path $\{y_j\}$, and each term in the sum in Eq.~\eqref{eq:feynmanpathsum} shares the same denominator $Q_0(\theta)$. Therefore, $\expval{C(\theta)}{0^n}$ is a degree $(O(m),O(m))$ rational function in $\theta$, and so is $\przero{C(\theta)}=|\expval{C(\theta)}{0^n}|^2$. $Q(\theta)$ is then the denominator of $\przero{C(\theta)}$. Finally, note that if $C_0$ only consists of 2-qubit gates, then $\przero{C(\theta)}$ is a degree $(8m,8m)$ rational function in $\theta$.

\end{proof}

Lemma~\ref{lemma:lowdegreerational} formally supports Property 2 of the $\theta$-Cayley perturbed random circuit distribution. Intuitively, this low-degree rational function structure allows us to apply polynomial interpolation techniques to obtain an approximation to the worst-case quantity, even though the worst-case point ($\theta=1$) is far from the average-case data points ($0<\theta\ll 1$). Details of this interpolation technique are presented in the proof of our main result (Theorem~\ref{thm:averagecasehardness}) and in Section~\ref{section:robustBW}. Before presenting the main result, it remains to establish Property 1 of the $\theta$-Cayley perturbed random circuit distribution, which is given in the following lemma.

\begin{lemma}[\cite{movassagh2020quantum}]\label{lemma:circuittvd}
Let $\mc H_{\mc A,\theta}$ be the $\theta$-Cayley perturbed random circuit distribution as in Definition~\ref{def:thetacayleyperturbation}, and $\mc H_{\mc A}$ be the distribution of Haar random circuits over $\mc A$. Then we have
\begin{equation}
    D_{\mathrm{TV}}(\mc H_{\mc A,\theta},\mc H_{\mc A})= O(m\theta),
\end{equation}
where $D_{\mathrm{TV}}(\cdot,\cdot)$ denotes the total variation distance between probability distributions and $m$ is the number of gates in $\mc A$.
\end{lemma}
\begin{proof}
This total variation distance can be bounded by considering the distribution of each individual gates. In~\cite{movassagh2020quantum}, it was shown that the total variation distance between the Cayley transformed random unitary $H(\theta)$ and Haar random unitary is $O(\theta)$ (also see Lemma~\ref{lemma:totalvariationdistance} in Appendix~\ref{appendix:cayleytransform}). Therefore by additivity of the total variation distance we have $D_{\mathrm{TV}}(\mc H_{\mc A,\theta},\mc H_{\mc A})=O(m\theta)$.
\end{proof}

Having established Property 1 and 2, we are now ready to state and prove our first result on the average-case hardness of random circuits.

\begin{theorem}\label{thm:averagecasehardness}
Let $\mc A$ be a circuit architecture so that computing $\przero{C}$ to within additive error $2^{-O(m)}$ is $\sharpp$-hard in the worst case. Then the following problem is $\sharpp$-hard under a $\BPP^{\NP}$ reduction: for any constant $\eta<\frac{1}{4}$, on input a random circuit $C\sim\mc H_{\mc A}$ with $m$ gates, compute the output probability $\przero{C}$ up to additive error $\delta=\exp\left(-O(m\log m)\right)$, with probability at least $1-\eta$ over the choice of $C$.
\end{theorem}

\begin{remark}
    Another way of stating Theorem~\ref{thm:averagecasehardness} is the following: suppose there exists an algorithm that belongs to some finite level of $\PH$ that, on input a random circuit $C\sim\mc H_{\mc A}$ with $m$ gates, computes the output probability $\przero{C}$ up to additive error $\delta=\exp\left(-O(m\log m)\right)$, with probability at least $1-\eta$ ($\eta<\frac{1}{4}$) over the choice of $C$ as well as the randomness of the algorithm. Then $\PH$ collapses to a finite level.
\end{remark}

\begin{proof}
Let $\mc O$ be an algorithm that correctly approximates $\przero{C}$ of a random circuit $C\sim\mc H_{\mc A}$ up to additive error $\delta$, with success probability at least $1-\eta$ over the choice of $C$. In the following, we show that there exists a $\BPP^{\NP^{\mc O}}$ procedure that on input \emph{any} circuit $C_0$, computes $\przero{C_0}$ up to additive error $\delta'=\delta\exp\left(O(m\log m)\right)$, with success probability at least $\frac{2}{3}$. The theorem statement then follows from the worst-case hardness of computing $\przero{C_0}$ over $\mc A$.

Consider any circuit $C_0$ with $m$ gates over the architecture $\mc A$. Create a new circuit $C_1$ as follows: for each gate $G_i$ ($i=1,\dots,m$) in $C_0$, we replace $G_i$ with $H_i G_i$, where $\{H_i\}$ is independently drawn from the Haar measure over the unitary group and has the same dimension as $G_i$. By the invariance of Haar measure, $C_1$ is distributed the same as $\mc H_{\mc A}$.

Fix the random unitary gates $\{H_i\}$. Next, we apply the $\theta$-Cayley perturbation on $C_0$ using the random seed $\{H_i\}$ as in Definition~\ref{def:thetacayleyperturbation} to get the perturbed circuit $C(\theta)$. By definition, we have $C(0)=C_1\sim \mc H_{\mc A}$ and $C(1)=C_0$.

By Lemma~\ref{lemma:lowdegreerational}, $\przero{C(\theta)}$ is a degree $(O(m),O(m))$ rational function in $\theta$. Let $P(\theta)$, $Q(\theta)$ be the numerator and denominator of $\przero{C(\theta)}$, respectively, then $\przero{C(\theta)}=\frac{P(\theta)}{Q(\theta)}$. Note that from the proof of Lemma~\ref{lemma:lowdegreerational}, we have that
\begin{equation}
    Q(\theta)=\prod_{j=1}^m \prod_l\left|1-i(1-\theta)\tan\frac{\varphi_{jl}}{2}\right|^2.
\end{equation}

The goal is to recover $\przero{C(1)}=\frac{P(1)}{Q(1)}$ from values of $\przero{C(\theta)}$ for small $\theta$. To do this, first note that $Q(\theta)$ is a known polynomial whose value can be computed in time $O(m)$ for any $\theta$. Therefore, the problem can be reduced to a polynomial interpolation for $P(\theta)$.

To analyze the error for the polynomial interpolation, it is useful to establish bounds for $Q(\theta)$. We can write $Q(\theta)$ as
\begin{equation}
    Q(\theta)=\prod_{j=1}^m \prod_l\left|1-i(1-\theta)\tan\frac{\varphi_{jl}}{2}\right|^2=\prod_{j=1}^m \prod_l\left(1+(1-\theta)^2\tan^2\frac{\varphi_{jl}}{2}\right),
\end{equation}
where it is easy to see that $Q(\theta)\geq 1$. Next, call the set of Haar random gates $\{H_j\}_{j=1,\dots,m}$ $\beta$-good, if all eigenvalues $\varphi_{jl}$ of all gates lie in the range $[-\pi+\beta,\pi-\beta]$. By Lemma~\ref{lemma:haardistributionrange}, this happens with probability at least $1-O(m\beta)$. Suppose we choose $\beta =O(m^{-1})$ such that $\{H_j\}$ is $\beta$-good with high constant probability. Then conditioned on $\{H_j\}$ being $\beta$-good, we have
\begin{nalign}
    Q(\theta)&=\prod_{j=1}^m \prod_l\left(1+(1-\theta)^2\tan^2\frac{\varphi_{jl}}{2}\right)\\
    &\leq \prod_{j=1}^m \prod_l\left(1+\tan^2\frac{\pi-\beta}{2}\right)\\
    &\leq \prod_{j=1}^m \prod_l\left(1+\frac{4}{\beta^2}\right)\\
    &= \left(1+O(m^2)\right)^{O(m)}\\
    &=\exp\left(O(m\log m)\right).
\end{nalign}
We also define the above upper bound of $Q(\theta)$ as $K$, where $K=\exp\left(O(m\log m)\right)$.

Next we apply the algorithm $\mc O$ on input $C(\theta)$. Notice that by definition, $\mc O$ works on inputs from the distribution $\mc H_{\mc A}$, while $C(\theta)$ is distributed according to $\mc H_{\mc A,\theta}$ as in Definition~\ref{def:thetacayleyperturbation}. Therefore, the success probability of $\mc O$ depends on the distance between the two distributions,
\begin{equation}
    \Pr_{C(\theta)\sim \mc H_{\mc A,\theta}}\left[\left|\mc O(C(\theta))-\przero{C(\theta)}\right|\geq \delta\right]\leq \eta +D_{\mathrm{TV}}(\mc H_{\mc A,\theta},\mc H_{\mc A}),
\end{equation}
where $D_{\mathrm{TV}}$ denotes total variation distance. By Lemma~\ref{lemma:circuittvd}, we have $D_{\mathrm{TV}}(\mc H_{\mc A,\theta},\mc H_{\mc A})=O(m\theta)$. Let $\Delta=O(m^{-1})$ and restrict $\theta$ in the interval $[0,\Delta]$, such that $D_{\mathrm{TV}}(\mc H_{\mc A,\theta},\mc H_{\mc A})=O(m\Delta)$ is upper bounded by a small constant.

Conditioned on $\mc O$ being successful and $\{H_j\}$ being $\beta$-good, we have $\left|\mc O(C(\theta))-\frac{P(\theta)}{Q(\theta)}\right|\leq \delta$. After seeing the output $\mc O(C(\theta))$ of $\mc O$, we multiply it with $Q(\theta)$ to get an estimate of $P(\theta)$, which satisfies
\begin{equation}
    \left|\mc O(C(\theta))Q(\theta)-P(\theta)\right|\leq \delta Q(\theta)\leq \delta K.
\end{equation}
By a simple union bound, the above equation holds with failure probability at most
\begin{nalign}\label{eq:oraclefailureprob}
    \Pr\left[\left|\mc O(C(\theta))Q(\theta)-P(\theta)\right|\geq \delta K\right]&\leq \Pr\left[\left|\mc O(C(\theta))-\przero{C(\theta)}\right|\geq \delta\vee \{H_j\}\text{ not }\delta\text{-good}\right]\\
    &\leq \eta+O(m\Delta)+O(m\beta)\leq\eta'<\frac{1}{4},
\end{nalign}
for a suitable choice of constants in the $O$-notation for $\Delta=O(m^{-1})$ and $\beta=O(m^{-1})$, such that $\eta'<\frac{1}{4}$.

To compute $P(1)$, we apply the algorithm $\mc O$ to a set of circuits $\{C(\theta_i)\}$, where $\theta_i$ ($i=1,\dots,O(m^2)$) is a set of equally spaced points in the interval $[0,\Delta]$, and different perturbed circuits $C(\theta_i)$ shares the same random seed $\{H_j\}$. By Eq.~\eqref{eq:oraclefailureprob}, we obtain a set of points $\{(\theta_i,y_i)\}$ such that
\begin{equation}
    \Pr\left[|y_i-P(\theta_i)|\geq \delta K\right]\leq\eta'<\frac{1}{4}.
\end{equation}
The problem is then reduced to a polynomial interpolation for $P(\theta)$, a degree $O(m)$ polynomial, with the noisy data $\{(\theta_i,y_i)\}$. Using our robust Berlekamp-Welch theorem which we develop in the next section (see Theorem~\ref{Thm:robustBW}), on input $\{(\theta_i,y_i)\}$ we can compute a number $p\approx P(1)$ with access to an $\NP$ oracle, where the error can be bounded by
\begin{nalign}
    |p-P(1)|&\leq \delta K \exp\left(O(m\log \Delta^{-1})+O(m)\right)\\
    &=\delta\cdot \exp(O(m\log m)).
\end{nalign}
Finally, notice that 
\begin{equation}
    \przero{C_0}=\przero{C(1)}=\frac{P(1)}{Q(1)}=P(1)
\end{equation}
as $Q(1)=1$, therefore by choosing $\delta=\exp\left(-O(m\log m)\right)$ with a sufficiently large constant, we can compute the worst-case output probability $\przero{C_0}$ up to additive error $\exp\left(-O(m\log m)\right)$. The overall procedure is in $\BPP^{\NP^{\mc O}}$. If $\mc O$ is an algorithm that belongs to some finite level of $\PH$, then this procedure also belongs to a finite level of $\PH$, and therefore by the worst-case $\sharpp$ hardness the $\PH$ collapses to a finite level.
\end{proof}

In the above proof we used a result developed in Section~\ref{section:robustBW} as a black box, which we refer to as robust Berlekamp-Welch algorithm. This algorithm allows us to do polynomial interpolation on noisy data, while tolerating a constant fraction of data points that can be arbitrarily wrong (see Section~\ref{section:techniques} for a detailed discussion). Similar to the Learning with Errors problem for solving noisy linear equations, the polynomial interpolation problem with both noise and corruption seems unlikely to be solved in polynomial time. However, recall that given the $\sharpp$ hardness of computing $\przero{C}$ in the worst case, in our worst-to-average-case reduction we are allowed to use any finite level of $\PH$. We show how to do such a polynomial interpolation in Section~\ref{section:robustBW} having access to a $\NP$ oracle. As a result, we are able to tolerate a constant failure probability in our average-case hardness.

Furthermore, our proof techniques can also be naturally applied to BosonSampling, where a fundamental question is to prove robust average-case hardness results for computing permanents of matrices with i.i.d. Gaussian entries~\cite{Aaronson2011linear,Aaronson2014Bosonsampling}. Consider a BosonSampling experiment with $n$ photons and $m=n^c$ ($c>2$) output modes. Then assuming no collision, the output probability corresponding to a output pattern $S=(s_1,\dots,s_m)$ ($s_i\in\{0,1\}$, $\sum_i s_i=n$) is
\begin{equation}
    \Pr[S]=\left|\Per(A_S)\right|^2,
\end{equation}
where $A$ is a $m\times n$ submatrix of a $m\times m$ Haar random unitary matrix, and $A_S$ is the $n\times n$ submatrix of $A$ whose rows are selected according to $S$. When $c$ is a large enough constant, $A_S$ is distributed close in total variation distance to matrices with i.i.d. complex Gaussian entries of mean 0 and variance $\frac{1}{m}$. We then obtain the following result using similar proof techniques:

\begin{corollary}\label{cor:permanent}
For any constant $\eta<\frac{1}{4}$, it is $\sharpp$-hard under a $\BPP^{\NP}$ reduction to compute the output probability of a $n$-photon $m=n^c$-mode BosonSampling experiment up to additive error $\delta=e^{-(c+4)n\log n-O(n)}$, with success probability $1-\eta$.
\end{corollary}

\begin{remark}\label{rmk:bosonsampling}
    As shown in \cite{Aaronson2011linear}, to prove the hardness of approximate sampling for BosonSampling experiments, it suffices to improve the constant in our result from $c+4$ to $c-1$. This is because on average the output probability is roughly
    \begin{equation}
        \frac{1}{\binom{m}{n}}\approx \frac{n!}{m^n}\approx e^{-(c-1)n\log n}.
    \end{equation}
    However, the barrier result shown by \cite{Aaronson2011linear} implies that this improvement cannot be obtained using similar techniques based on polynomial interpolation. See Appendix~\ref{sec:barrier} for a detailed discussion, where we also show that this barrier result does not rule out improving the constant to $c+1$.
\end{remark}

\begin{proof}
Let $\mc G^{n\times n}$ denote the distribution over $n\times n$ matrices where each entry is independently distributed according to $\mc{CN}(0,1)$, the standard complex Gaussian distribution. Suppose $c$ is a large enough constant\footnote{The proof in \cite{Aaronson2011linear} requires $c$ to be at least 5, and it was conjectured that $c>2$ suffices.}, then an algorithm for approximating the output probability of a BosonSampling experiment can be used to approximate
\begin{equation}
    p_X:=\frac{\left|\Per(X)\right|^2}{m^n},\,\,\,\,X\sim \mc G^{n\times n}
\end{equation}
with a small loss in the success probability.

Let $\mc O$ be an algorithm that, on input a random matrix $X\sim\mc G^{n\times n}$, correctly approximates $p_X$ up to additive error $\delta$, with success probability at least $1-\eta$ over the choice of $X$. In the following, we show that there exists a $\BPP^{\NP^{\mc O}}$ procedure that on input \emph{any} $0/1$ matrix $X_0$, computes $\left|\mathrm{Per}(X_0)\right|^2$ up to additive error $\delta'=\delta\exp\left((c+4)n\log n+O(n)\right)$, with success probability at least $\frac{2}{3}$. The theorem statement then follows from the worst-case $\sharpp$ hardness of computing $\mathrm{Per}(X_0)$.

Let $X_1\sim\mc G^{n\times n}$ and define
\begin{equation}
    X(\theta):=(1-\theta)X_1+\theta X_0.
\end{equation}
Then $\left|\mathrm{Per}(X(\theta))\right|^2$ is a degree $d=2n$ polynomial in $\theta$. By a result of \cite{Aaronson2011linear}, the total variation distance between the distribution of $X(\theta)$ and $\mc G^{n\times n}$ is $O(n^2\theta)$. This can be shown by calculating the distance for one entry, which is $O(\theta)$, and then multiply by $n^2$ as the entries are independently distributed.

Consider $O(n^2)$ uniformly spaced points $\{\theta_i\}$ in $[0,\Delta]$ with $\Delta=O(n^{-2})$. For a suitable choice of constants, we can guarantee that for each data point $\theta_i$,
\begin{equation}
    \Pr[\left|\mc O(X(\theta_i))-\frac{\left|\Per(X(\theta_i))\right|^2}{m^n}\right|\geq\delta]\leq \eta+O(n^2\Delta)\leq\eta'
\end{equation}
for some constant $\eta'<\frac{1}{4}$. The procedure then follows by sending the points $\{(\theta_i,\mc O(X(\theta_i)))\}$ to the robust Berlekamp-Welch algorithm (Theorem~\ref{Thm:robustBW}) to obtain the desired approximation to the worst-case quantity $\left|\mathrm{Per}(X_0)\right|^2=\left|\mathrm{Per}(X(1))\right|^2$. Overall we have a $\BPP^{\NP^{\mc O}}$ procedure for computing $\left|\mathrm{Per}(X_0)\right|^2$ up to additive error
\begin{equation}
    \delta\cdot m^n\cdot \exp\left(d\log\Delta^{-1}+O(d)\right)=\delta\cdot \exp\left((c+4)n\log n+O(n)\right),
\end{equation}
which concludes the proof.
\end{proof}

\subsection{Noisy circuits and error detection}\label{sec:noisyworstcase}
Next, we generalize our average-case hardness result to noisy random quantum circuits. We first define the computational problem of approximating the output probability of noisy quantum circuits, and show that worst-case hardness can be established via error detection. Then in Section~\ref{sec:noisyaveragecasehardness} we show that our previous reduction can be directly applied to noisy circuits, therefore obtaining average-case hardness.

For simplicity, here we consider noise models that are \emph{local} and \emph{stochastic}. Namely, for a circuit $C$ with $m$ gates, there are $r=O(m)$ independent noise channels, where each noise channel $\mc N_k$ ($k=1,\dots,r$) acts on a constant number of neighboring qubits and has the following form,
\begin{equation}\label{eq:noisemodel}
    \mc N_k=(1-\gamma)\mc I+\gamma \mc E_k,
\end{equation}
where $\mc I$ is the identity channel, $\mc E_k$ is an arbitrary CPTP map, and $\gamma$ is the noise rate. We assume that the noise model $\mc N=\{\mc N_k\}_{k=1,\dots,r}$ is fixed \emph{a priori} and does not depend on the choice of gates. In principle, our hardness proof can be applied to more general noise models, as long as they are gate-independent and feasible for error correction/detection. In particular, the choice of a fixed noise rate for each noise channel is not essential and is for simplicity, as otherwise we can define $\gamma$ as the maximum of all noise rates. 

Our noise model \eqref{eq:noisemodel} can be understood as follows: for each $k=1,\dots,r$, with probability $1-\gamma$, do nothing; with probability $\gamma$, apply $\mc E_k$. The overall probability of having no error is at least $(1-\gamma)^r$. A simulation of the noisy circuit needs to take expectation over these stochastic paths. Equivalently, we can think of $\mc N_k$ as quantum channels that deterministically acts on input density matrices. Simulating the output probability of a noisy circuit is therefore equivalent to computing the diagonal elements of the output density matrix of the noisy circuit.

\begin{definition}
Fix a noise model $\mc N$. An algorithm simulates a circuit $C$ of $n$ qubits under the noise model $\mc N$ up to additive error $\delta$, if on input $C$, it computes a number $p$ that satisfies
\begin{equation}
    \left|p-\przero{C,\mc N}\right|\leq \delta,
\end{equation}
where $\przero{C,\mc N}:=\Tr\left[\ketbra{0^n}\cdot\mc C_{\mc N}(\ketbra{0^n})\right]$, and $\mc C_{\mc N}$ denotes the quantum channel induced by applying $\mc N$ to $C$. We use $\przero{C}$ to denote the same output probability in the noiseless case.
\end{definition}

To show the computational hardness of simulating noisy circuits, one straightforward idea is to apply fault-tolerant quantum error correction. Suppose $\mc A$ is an architecture over which fault-tolerant error correction is possible. For example, $\mc A$ can consist of nearest neighbor two qubit gates over a 1D or 2D grid. Let $C_0$ be a quantum circuit whose output probability is $\sharpp$-hard to approximate within an exponentially small additive error. Then by fault-tolerance threshold theorems~\cite{Kitaev_1997,aharonov2008fault}, when the noise rate of $\mc N$ is below a constant threshold, we can construct a quantum circuit $C$ over $\mc A$ with a polynomial overhead in the size of $C_0$, such that even with the presence of noise, the output distribution of $C$ is exponentially close to the output distribution of $C_0$ in total variation distance. The computational hardness of simulating $C$ with the noise model $\mc N$ then follows from the hardness of simulating $C_0$.

An improvement of the above argument developed by Fujii~\cite{fujii2016noise} is to identify the fact that fault-tolerant \emph{error detection} suffices. This has the benefit of having a smaller overhead in implementing fault-tolerance, and more importantly, a higher noise threshold~\cite{reichardt2006error,aliferis2007accuracy}. Fujii's result, which establishes the hardness of sampling from the output distribution of noisy circuits in the worst case, follows from a similar post-selection argument by Bremner, Jozsa and Shepherd~\cite{Bremner2011classical}: following Aaronson's $\PostBQP=\PP$ result~\cite{Aaronson2005quantum}, to prove that a family of circuits are hard to sample from, it suffices to show that they are universal for quantum computation under post-selection. To apply this argument, note that applying standard threshold theorems for an error detecting code, we have that when post-selecting the ``no error detected" event, the output distribution of a noisy circuit that implements error detection is exponentially close to the underlying logical circuit in total variation distance, which means that the error-detected noisy circuits are universal for $\BQP$ under post-selection. Therefore, when noise is below the error detection threshold, the output distribution of noisy circuits are hard to sample from in the worst case.

Different from Fujii's result, here we consider the worst-case hardness of computing the output probability of noisy circuits. We show that here error detection also suffices: when the noise rate is below the error detection threshold, computing the output probability up to an exponentially small additive error is hard for the Polynomial Hierarchy.

\begin{theorem}\label{thm:mainworstcasehardness}
Fix a noise model $\mc N$, let $\gamma_{\mathrm{th}}$ be the corresponding fault-tolerant error detection threshold, and suppose the noise rate of $\mc N$ satisfies $\gamma<\gamma_{\mathrm{th}}$. Suppose there exists a $\PH$ procedure (that belongs to some finite level of $\PH$) that on input $C$ with $m$ gates, computes $\przero{C,\mc N}$ up to additive error $\delta=\exp(-O(m\log m))$. Then $\PH$ collapses to a finite level.
\end{theorem}

The proof follows from a similar application of the threshold theorem, and is presented in Appendix~\ref{appendix:worstcasehardness}.

\subsection{Average-case hardness of noisy circuits}\label{sec:noisyaveragecasehardness}
We prove average-case hardness of noisy circuits by extending our worst-to-average-case reduction to the noisy case. We start by re-examining the two properties of the $\theta$-Cayley perturbed random circuit distribution in the noisy setting. First, note that Property 1 is unaffected, as we have a fixed noise model that is independent of the choice of gates. Second, in the following we show that Property 2 still holds, which follows from the linearity of quantum channels and the convex structure of the Feynman path integral.

\begin{lemma}\label{lemma:noisylowdegreerational}
For any circuit $C_0$ and any fixed noise model $\mc N$, let $C(\theta)$ be a circuit from the $\theta$-Cayley perturbed random circuit distribution as in Definition~\ref{def:thetacayleyperturbation}. Then $\przero{C(\theta),\mc N}$ is a degree $(O(m),O(m))$ rational function in $\theta$.
\end{lemma}
\begin{proof}
The proof follows a similar argument as in Lemma~\ref{lemma:lowdegreerational}. Consider a noise model $\mc N=\{\mc N_j\}$. We prove the lemma statement in a special case where $r=m$ and each gate $H_j(\theta)G_j$ is followed by a noise channel $\mc N_j$, which can be directly generalized to the general case where the noise channels are arbitrarily placed.

As each noise channel $\mc N_j$ only acts on a constant number of qubits, it can be written as $\mc N_j(\rho)=\sum_{l_j}E_{jl_j}\rho E_{jl_j}^\dag$ where the summation is over constant number of terms, and $\{E_{jl_j}\}_{l_j}$ are the noise operators of $\mc N_j$ satisfying $\sum_{l_j}E_{jl_j}^\dag E_{jl_j}=I$. Let $\mc U_j$ be the unitary channel of the gate $H_j(\theta)G_j$.  Then, we have
\begin{nalign}
    \przero{C(\theta),\mc N}&=\Tr\left[\ketbra{0^n}\cdot \mc N_m\circ \mc U_m\circ\cdots\circ \mc N_1\circ \mc U_1(\ketbra{0^n}) \right]\\
    &=\sum_{l_1,\dots,l_m}\left|\expval{\prod_{j=1}^m E_{jl_j}H_j(\theta)G_j}{0^n}\right|^2.
\end{nalign}
Here for simplicity we omit the identity operators acting on qubits outside the support of the gates. The summation index $l_1,\dots,l_m$ can be understood as a \emph{noise trajectory}, where the noise operators are stochastically applied. Note that for each term in the above summation, the same analysis for $\przero{C(\theta)}$ can be applied (cf. Lemma~\ref{lemma:lowdegreerational}), which shows that $\left|\expval{\prod_{j=1}^m E_{jl_j}H_j(\theta)G_j}{0^n}\right|^2$ is a degree $(O(m),O(m))$ rational function in $\theta$ with denominator $Q(\theta)=\prod_{j=1}^m \prod_l\left|1-i(1-\theta)\tan\frac{\varphi_{jl}}{2}\right|^2$. Note that this denominator does not depend on the noise trajectory. Therefore, each term in the summation shares the same denominator, and the sum $\przero{C(\theta),\mc N}$ is also a degree $(O(m),O(m))$ rational function in $\theta$ with denominator $Q(\theta)$.
\end{proof}

Having established worst-case hardness of noisy circuits (Theorem~\ref{thm:mainworstcasehardness}) and the low-degree rational function property (Lemma~\ref{lemma:noisylowdegreerational}) for the reduction, the proof of average-case hardness then follows from a similar argument as in Theorem~\ref{thm:averagecasehardness}.

\begin{theorem}\label{thm:noisyaveragecasehardness}
Fix a noise model $\mc N$. Let $\mc A$ be a circuit architecture that can implement universal fault-tolerant error detection, and suppose the noise rate $\gamma<\gamma_{\mathrm{th}}$ is smaller than the error detection threshold. Suppose there exists a $\PH$ procedure (that belongs to some finite level of $\PH$) that, on input a random circuit $C\sim\mc H_{\mc A}$ with $m$ gates, computes the output  probability $\przero{C,\mc N}$ up to additive error $\delta=\exp\left(-O(m\log m)\right)$, with probability at least $1-\eta$ ($\eta<\frac{1}{4}$) over the choice of $C$ as well as the randomness of the algorithm. Then $\PH$ collapses to a finite level.
\end{theorem}

\begin{proof}
Let $\mc O$ be an algorithm that correctly approximates $\przero{C,\mc N}$ of a random circuit $C\sim\mc H_{\mc A}$ up to additive error $\delta$, with success probability at least $1-\eta$ over the choice of $C$. In the following, we show that there exists a $\BPP^{\NP^{\mc O}}$ procedure that on input \emph{any} circuit $C_0$, computes $\przero{C_0,\mc N}$ up to additive error $\delta'=\delta\exp\left(O(m\log m)\right)$, with success probability at least $\frac{2}{3}$. The theorem statement then follows from Theorem~\ref{thm:mainworstcasehardness}, the worst-case hardness result.

Following the same reduction technique as in Theorem~\ref{thm:averagecasehardness}, we sample a set of random unitary gates $\{H_i\}$ and apply the $\theta$-Cayley perturbation on $C_0$ using the random seed $\{H_i\}$ as in Definition~\ref{def:thetacayleyperturbation} to get the perturbed circuit $C(\theta)$. By definition, we have $C(0)=C_1\sim \mc H_{\mc A}$ and $C(1)=C_0$.

By Lemma~\ref{lemma:noisylowdegreerational}, $\przero{C(\theta),\mc N}$ is a degree $(O(m),O(m))$ rational function in $\theta$. Let $P(\theta)$, $Q(\theta)$ be the numerator and denominator of $\przero{C(\theta),\mc N}$, respectively, then $\przero{C(\theta),\mc N}=\frac{P(\theta)}{Q(\theta)}$. Note that from the proof of Lemma~\ref{lemma:noisylowdegreerational}, we have that $Q(\theta)=\prod_{j=1}^m \prod_l\left|1-i(1-\theta)\tan\frac{\varphi_{jl}}{2}\right|^2$ with $Q(1)=1$.

The goal is to recover $\przero{C_0,\mc N}$ from values of $\przero{C(\theta),\mc N}$ for small $\theta$. To do this, first note that $Q(\theta)$ is a known polynomial whose value can be computed in time $O(m)$ for any $\theta$. Therefore, the problem can be reduced to a polynomial interpolation for $P(\theta)$, as 
\begin{equation}
    \przero{C_0,\mc N}=\przero{C(1),\mc N}=\frac{P(1)}{Q(1)}=P(1).
\end{equation}
The rest of the proof is therefore the same as in Theorem~\ref{thm:averagecasehardness}.
\end{proof}

\begin{figure}[t]
\tikzset{phase/.append style={blue}}
 \centering
 \begin{subfigure}[b]{0.49\textwidth}
     \centering
     	\begin{quantikz}[column sep=0.25cm,row sep=0.2cm]
		\lstick{$\ket{0}$} & \gate[2]{} & \phase{} & \qw 		& \phase{} & \gate[2]{} & \phase{} & \qw		& \phase{} & \meter{} \\
		\lstick{$\ket{0}$} & 			& \phase{} & \gate[2]{} & \phase{} &			& \phase{} & \gate[2]{} & \phase{} & \meter{} \\
		\lstick{$\ket{0}$} & \gate[2]{} & \phase{} & 			& \phase{} & \gate[2]{} & \phase{} &			& \phase{} & \meter{} \\
		\lstick{$\ket{0}$} & 	 		& \phase{} & \gate[2]{} & \phase{} &			& \phase{} & \gate[2]{} & \phase{} & \meter{} \\
		\lstick{$\ket{0}$} & \qw 		& \phase{} & 			& \phase{} & \qw	    & \phase{} & 			& \phase{} & \meter{}
		\end{quantikz}
     \caption{}
     \label{fig:localrandom}
 \end{subfigure}
 \begin{subfigure}[b]{0.49\textwidth}
     \centering
     	\begin{quantikz}[column sep=0.25cm,row sep=0.2cm]
		\lstick{$\ket{0}$} & \gate[5]{} & \phase{} & \gate[5]{} & \phase{} & \gate[5]{} & \phase{} & \gate[5]{}	& \phase{} & \meter{} \\
		\lstick{$\ket{0}$} & 			& \phase{} & 			& \phase{} &			& \phase{} &			& \phase{} & \meter{} \\
		\lstick{$\ket{0}$} & 			& \phase{} & 			& \phase{} & 			& \phase{} &			& \phase{} & \meter{} \\
		\lstick{$\ket{0}$} & 	 		& \phase{} & 			& \phase{} &			& \phase{} & 			& \phase{} & \meter{} \\
		\lstick{$\ket{0}$} & 	 		& \phase{} & 			& \phase{} & 		    & \phase{} & 			& \phase{} & \meter{}
		\end{quantikz}
     \caption{}
     \label{fig:globalrandom}
 \end{subfigure}
    \caption{Noisy random circuit sampling, where each white box is a unitary gate independently drawn from the Haar measure, and blue dots represent noise. (a) Experimental implementation of random circuit sampling, where the circuit consists of nearest-neighbor two-qubit gates in an alternating architecture. (b) A toy model for studying the convergence to uniform distribution, where each layer is one global random unitary gate.}
    \label{fig:randomcircuitsampling}
\end{figure}
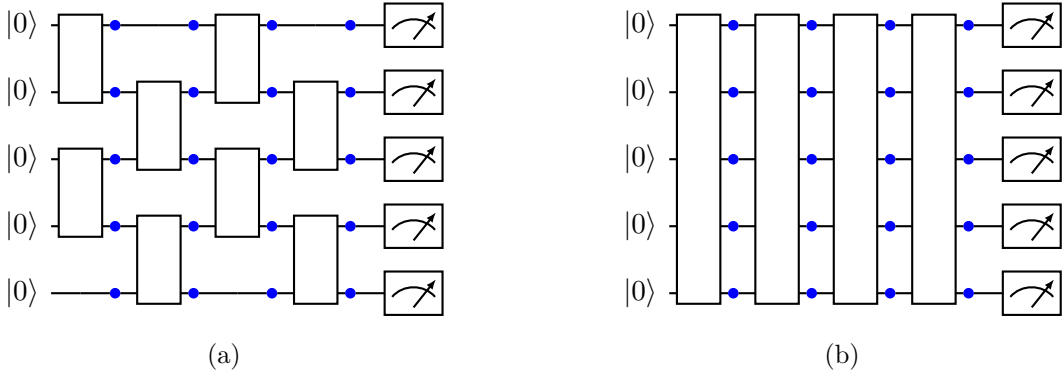

An important distinction between the noisy random circuit model and the ideal model brings interesting implications to our hardness result. In noisy random circuits, it was shown~\cite{aharonov1996limitations,gao2018efficient} that the output distribution converges to the uniform distribution very quickly as circuit depth increases. Therefore, trivially outputs $\frac{1}{2^n}$ is already a good approximation algorithm for $\przero{C,\mc N}$, while this is not the case for ideal random circuits, whose output probability has the Porter-Thomas distribution which is far from uniform in total variation distance. In other words, the amount of additive imprecision that we can tolerate in the hardness proof has an \emph{upper bound}: it cannot be larger than the distance between the noisy output distribution and the uniform distribution.

This behavior can be understood as follows: the decoherence caused by noise leaks quantum information to the environment (here we consider incoherent noise models), and the amount of information that remains in the system decreases as circuit depth increases, which eventually goes to 0 and corresponds to the maximally mixed state. It was first shown by Aharonov \textit{et al.}~\cite{aharonov1996limitations} that the output distribution of $\emph{any}$ circuit with i.i.d. depolarizing noise converges to the uniform distribution with rate $e^{-\gamma d}$, where $\gamma$ is the noise rate on each qubit at a time step, and $d$ is circuit depth. Gao and Duan~\cite{gao2018efficient} proved a similar convergence speed for random circuits with general i.i.d. Pauli noise. However, it was conjectured~\cite{boixo2018characterizing} that the actual convergence speed for noisy random circuits is much faster, with rate $e^{-\gamma nd}$ (or $e^{-\gamma O(m)}$), which was indeed observed in experiments~\cite{arute2019quantum}. While the proof of this behavior for the local random circuit model (Fig.~\ref{fig:randomcircuitsampling}a) is out of the scope for this paper, we conjecture that it is true at least for some parameter region of $(\gamma,n,d)$. Furthermore, we provide a rigorous analysis for the toy model with global random gates (Fig.~\ref{fig:randomcircuitsampling}b), which we present in Appendix~\ref{appendix:noisyconvergence}.

\begin{lemma}\label{lemma:noisyconvergence}
The output distribution of the global noisy random circuit model (Fig.~\ref{fig:randomcircuitsampling}b) converges to the uniform distribution with speed $e^{-\Omega(\gamma nd)}$. More specifically, let $p_e$ denote the output distribution of the noisy circuit with i.i.d. depolarizing noise with noise rate $\gamma$. Then
\begin{equation}
    \E\left[2^n \sum_x p_{e}^2(x)-1\right]\leq\exp(-c\gamma n (d-d_0)),
\end{equation}
where $c,d_0>0$ are universal constants, and the expectation is over the global random circuit ensemble.
\end{lemma}

To apply this result, we can use a simple Markov's inequality which suggests that the output distribution of \emph{most} noisy quantum circuits satisfy $2^n \sum_x p_{e}^2(x)-1=e^{-\Omega(\gamma n d)}$. Then, for each of these circuits, the total variation distance between the noisy distribution and the uniform distribution can be bounded by
\begin{equation}
    D_{\mathrm{TV}}(p_e,p_{\mathrm{uniform}})=\frac{1}{2}\sum_{x}\left|p_e(x)-\frac{1}{2^n}\right|\leq \frac{1}{2}\sqrt{2^n\sum_x p_e^2(x)-1}
\end{equation}
via Cauchy–Schwarz inequality.

An immediate observation is that this convergence to uniformity behavior implies that our hardness result in the noisy case is essentially optimal. Indeed, for a family of noisy random circuits where the distance to uniform is bounded by $e^{-O(m)}$, our hardness result says that it is still hard to compute the output probability to within $e^{-O(m\log m)}$ additive imprecision, while it is trivial to achieve $e^{-O(m)}$ additive imprecision. The key point here is that our proof is noise and architecture agnostic, as long as they enables error detection. Therefore, it is impossible to improve our hardness result to have better robustness, as long as the agnostic property remains.

The same limitation also applies to the noiseless case. In our proof, the noisy hardness result is a direct generalization of the noiseless case due to the linearity and convex structure, which means that our proof for the noiseless case cannot be improved either, as long as it can still be generalized to the noisy case. Therefore, any attempt to address the hardness of approximate sampling via a reduction to computing output probabilities (which requires a robustness of $2^{-n}/\poly(n)$) in the ideal case cannot be natually adapted to noisy circuits. In this sense our result can be viewed as a complement of \cite{napp2020efficient}, which shows that the proof technique needs to be sensitive to depth because of an efficient simulation algorithm for shallow random circuits.

\section{Robust Berlekamp-Welch}\label{section:robustBW}

We prove the following theorem for robust polynomial interpolation. All polynomials considered here have real coefficients, and our results can be generalized to the complex case, as for polynomials with complex coefficients we can deal with the real and imaginary parts separately. Below our theorem is stated with success probability $\frac{2}{3}$, which can be amplified to $1-1/\expp(d)$ by taking the median of $\poly(d)$ independent experiments.

\begin{theorem}[Robust Berlekamp-Welch]\label{Thm:robustBW}
For a degree $d$ polynomial $P(x)$, suppose there is a set of data points $D=\{(x_i,y_i)\}$ such that $|D|=100d^2$ and $\{x_i\}$ is equally spaced in the interval $[0,\Delta]$ ($\Delta<1$). Furthermore, assume that each point $(x_i,y_i)$ satisfies
\begin{equation}\label{eq:datapoint}
    \Pr\left[\left|y_i-P(x_i)\right|\geq \delta\right]\leq\eta,
\end{equation}
where $\eta<\frac{1}{4}$ is a constant, then there exists a $\P^{\NP}$ algorithm which, takes $D$ as input, returns a number $p$ such that $|p-P(1)|\leq \delta e^{d\log \Delta^{-1} + O(d)}$, with success probability at least $\frac{2}{3}$.
\end{theorem}

\begin{remark}\label{rmk:rakhmanov}
    Note that our result can be improved to using only $O(d)$ data points, by applying a result of Rakhmanov~\cite{Rakhmanov2007bounds}. However, the number of data points we use does not affect our main results, as long as it is polynomial in $d$. Here we give a simpler proof using $O(d^2)$ data points.
\end{remark}

\begin{proof}
For simplicity, our proof is presented by specifying $\eta=\frac{1}{300}$, which can be naturally generalized to any $\eta<\frac{1}{4}$. Say a point $(x_i,y_i)$ is correct if $|y_i-P(x_i)|\leq \delta$. By Eq.~\eqref{eq:datapoint}, the expected number of wrong points is less than $\frac{1}{300}$ fraction. By Markov's inequality, with probability at least $\frac{2}{3}$, $0.99$ fraction of points in $D$ are correct. In the following we show that conditioned on $0.99$ fraction of points being correct (denote the set of correct points as $F$), there exists a deterministic $\P^{\NP}$ algorithm that computes a number $p$ that satisfies $|p-P(1)|\leq \delta e^{d\log \Delta^{-1} + O(d)}$. This implies the statement of the theorem.

Consider the following computational problem: 

\begin{problem}\label{prob:interpolation}
    Given $100d^2$ points $D$, decide if there exists at least $0.99$ fraction of them (denoted as $F'$) and a degree $d$ polynomial $Q(x)$ that satisfy $|y_i-Q(x_i)|\leq \delta,\forall(x_i,y_i)\in F'$.
\end{problem}
Problem~\ref{prob:interpolation} is in $\NP$ because a certificate $(F',Q)$ can be efficiently verified by checking each point in $F'$.

When giving $D$ in the statement of the theorem as input, this problem has a satisfying certificate $(F,P)$. Having access to a $\NP$ oracle, as this problem is guaranteed to have a solution, we can find a certificate $(F',Q)$ by performing binary search in the certificate space, which may be different from $(F,P)$. However, in the following we show that any solution will satisfy the requirement of the theorem, and the algorithm simply outputs $Q(1)$.

Let $R(x)=P(x)-Q(x)$. As $|F\cap F'|\geq 0.98|D|$, there are at least 0.98 fraction of points in $D$ that satisfies
\begin{equation}
    |R(x_i)|\leq |P(x_i)-y_i|+|Q(x_i)-y_i|=2\delta.
\end{equation}
Recall that $D$ is a set of equally spaced points in $[0,\Delta]$. Lemma~\ref{Lemma:zeph} says that in such condition $R(x)$ can be uniformly bounded in the interval $[0,\Delta]$:
\begin{equation}\label{eq:firstmissing}
    |R(x)|\leq \delta e^{O(d)},\,\,\,\,\forall x\in[0,\Delta].
\end{equation}
Finally, as $R$ is uniformly bounded in $[0,\Delta]$, we can use Lemma~\ref{lemma:longdistanceextrapolation} to bound the error at $x=1$,
\begin{equation}\label{eq:secondmissing}
    |R(1)|=|P(1)-Q(1)|\leq \delta e^{O(d)}e^{d\log\Delta^{-1} + d\log 8}=\delta e^{d\log \Delta^{-1} + O(d)},
\end{equation}
which concludes the proof.
\end{proof}

The above proof needs several additional results which we develop in the following. First, recall the following result of Paturi, which gives a bound of how fast a polynomial can grow when it is uniformly bounded in a small interval.

\begin{lemma}[Paturi~\cite{Paturi1992on}]\label{Lemma:Paturi}
Let $P$ be a degree $d$ polynomial which satisfies $|P(x)|\leq \varepsilon$ for $x\in[0,\Delta]$ ($\Delta<1$). Then we have $|P(1)|\leq \varepsilon e^{4d\Delta^{-1}}$.
\end{lemma}
\begin{proof}
This is implied by Fact 2 and Corollary 2 of \cite{Paturi1992on}.
\end{proof}

The first missing step in the proof of Theorem~\ref{Thm:robustBW}, which is in Eq.~\eqref{eq:firstmissing}, is to show the following: if a degree $d$ polynomial is bounded by $\delta$ for a constant fraction of a set of uniformly spaced points in $[0,\Delta]$, then the maximum value of the polynomial in the interval $[0,\Delta]$ is at most $\delta e^{O(d)}$. We show this in the following lemma, which is the main technical part of the full proof.

\begin{lemma}\label{Lemma:zeph}
Let $P$ be a degree $d$ polynomial which satisfies $|P(x)|\leq 1$ for any $x\in A$, where $A\subseteq B$ and $B$ is a set of equally spaced points in $[0,1]$ such that $|B|=100d^2$ and $|A|\geq 0.98 |B|$. Then $P$ satisfies $|P(x)|\leq 2^{O(d)},\forall x\in[0,1]$.
\end{lemma}
\begin{proof}
Denote $A$ as good points at which $P$ is bounded, and denote $B-A$ as bad points. We first prove the bound under a special case, where the good points are exactly the first 0.98 fraction. More specifically, for the set of uniformly spaced points
\begin{equation}
    B=\{0,\delta,2\delta,\dots,(100d^2-1)\delta\}\subseteq [0,1]
\end{equation}
where $\delta=\frac{1}{100d^2}$, suppose the polynomial is bounded at the first 0.98 fraction of points,
\begin{equation}
    |P(x)|\leq 1,\,\,\,\,\forall x\in A=\{0,\delta,\dots,(98d^2-1)\delta\}.
\end{equation}
When a polynomial is bounded at consecutive uniformly spaced points in an interval, it can also be uniformly bounded at all points in that interval. Using Lemma~\ref{lemma:bernsteinmarkov} we get
\begin{equation}
    |P(x)|\leq O(1),\,\,\,\,\forall x\in [0,(98d^2-1)\delta].
\end{equation}
As $P(x)$ is uniformly bounded on a constant subinterval of $[0,1]$, we can then use Lemma~\ref{Lemma:Paturi} to obtain the desired bound $|P(x)|\leq 2^{O(d)},\forall x\in[0,1]$. Also note that both inequalities are saturated by Chebyshev polynomial, so the bound is tight.

We proceed by proving the bound for the general case where $A$ can be arbitrarily distributed in $B$. Let $M=\max_{x\in[0,1]}|P(x)|$ and suppose the maximum is achieved at $x^*\in[\frac{1}{2},1]$. Consider the roots of $P$
\begin{equation}
    P(x)=L\prod_{k=1}^d (x-r_k),\,\,\,\,L>0,
\end{equation}
where $r_k$ can have real and imaginary parts. We construct a new polynomial $Q(x)=L\prod_{k=1}^d (x-r_k')$ according to the following rule:
\begin{equation}
    r_k'=\begin{cases}r_k,&\Re(r_k)\leq x^*,\\
    2x^*-\Re(r_k)+i\Im(r_k),&\Re(r_k)> x^*.\end{cases}
\end{equation}
Let $M'=\max_{x\in[0,1]}|Q(x)|$. We show that this construction is useful because $Q(x)$ has the following properties:
\begin{enumerate}
    \item $\deg(Q)=\deg(P)=d$.\\
    \emph{Proof.} By construction.
    \item $M'\geq M$.\\
    \emph{Proof.} $M'\geq |Q(x^*)|=|P(x^*)|=M$.
    \item $|Q(x)|\leq |P(x)|,\forall x\in[0,x^*]$.\\
    \emph{Proof.} This follows because $|x-r_k'|\leq |x-r_k|,\forall k,\forall x\in[0,x^*]$.
    \item $M'=|Q(1)|$.\\
    \emph{Proof.} $|Q(x)|$ is monotonic increasing at $[x^*,1]$ because $\forall x^*\leq x<y\leq 1$, we have $|x-r_k'|\leq |y-r_k'|$. Also, $\forall x\in[0,x^*],|Q(x)|\leq |P(x)|\leq M=|Q(x^*)|$. This suggests that the maximum occurs at $x=1$.
\end{enumerate}
According to property 2, we proceed by proving an upper bound for $M'$. Let $B'=\{0,\delta,2\delta,\dots,(50d^2-1)\delta\}\subseteq[0,\frac{1}{2}]$. Among these points, at least $0.96$ fraction of them are good for $Q(x)$, denoted by $A'$. More specifically,
\begin{equation}
    |Q(x)|\leq |P(x)|\leq 1,\forall x\in A'\subseteq B'.
\end{equation}
If $A'$ is a consecutive set starting from 0 (i.e. $A'=\{0,\delta,2\delta,\dots,(|A'|-1)\delta\}$), then we can obtain the desired bound by repeating the argument at the beginning of this proof. When this is not the case, there exists $y\in B'$ such that $y$ is bad and $y+\delta$ is good. We proceed by converting $Q(x)$ to a new polynomial to eliminate this kind of points, while maintaining the good properties of $Q$.

Let $y$ be the maximum element in $B'$ such that $y$ is bad and $y+\delta$ is good (assuming such $y$ exists). By definition, all good points larger than $y$ is consecutive, denoted as
\begin{equation}
    A'\cap[y,\frac{1}{2}]=\{y+\delta,y+2\delta,\dots,y+T\delta\}
\end{equation}
for some $T\geq 1$. We perform the following operation on $Q(x)=L\prod_{k=1}^d (x-r_k)$, resulting in a new polynomial $R(x)=L\prod_{k=1}^d (x-r_k')$:
\begin{equation}
    r_k'=\begin{cases}r_k,&\Re(r_k)\leq y,\\
    \Re(r_k)-\delta+i\Im(r_k),&\Re(r_k)> y.
    \end{cases}
\end{equation}
We prove the following properties for $R$.
\begin{enumerate}\addtocounter{enumi}{4}
    \item $|R(x-\delta)|\leq 1,\forall x\in A'\cap[y,\frac{1}{2}]$.\\
    \emph{Proof.} $\forall x\in A'\cap[y,\frac{1}{2}]$, notice that $|x-\delta-r_k'|=|x-r_k|$ when $\Re(r_k)> y$, and $|x-\delta-r_k'|\leq|x-r_k|$ when $\Re(r_k)\leq y$. Therefore $|R(x-\delta)|\leq |Q(x)|\leq 1$.
    \item $|R(x)|\leq 1,\forall x\in A'\cap[0,y]$.\\
    \emph{Proof.} $\forall x\in A'\cap[0,y]$, notice that $|x-r_k'|=|x-r_k|$ when $\Re(r_k)\leq y$, and $|x-r_k'|\leq|x-r_k|$ when $\Re(r_k)> y$. Therefore $|R(x)|\leq |Q(x)|\leq 1$.
\end{enumerate}
From property 5-6 we conclude that $R(x)$ has the same number of good points as $Q(x)$ (if some point outside of $(A'\cap [0,y])\cup\{y,y+\delta,\dots,y+(T-1)\delta\}$ is good for $R$, we still label it as bad). As $y$ is good for $R$, we also have $\max\{z\in B':z+\delta\in A'\}<y$. Finally we show that $R$ has a larger maximum.
\begin{enumerate}\addtocounter{enumi}{6}
    \item Let $M''=\max_{x\in[0,1]}|R(x)|$. Then $M''\geq M'$.\\
    \emph{Proof.} As $|1-r_k'|\geq |1-r_k|$, we have $M''\geq |R(1)|\geq |Q(1)|=M'$.
\end{enumerate}
Next, we repeat the above process until $\{z\in B':z+\delta\in A'\}$ is empty, and denote the resulting polynomial as $\tilde{R}$, for which property 5-7 still holds. As $\tilde{R}$ is bounded by 1 at $\{0,\delta,2\delta,\dots,(48d^2-1)\delta\}$, we obtain
\begin{equation}
    M'\leq \max_{x\in[0,1]}|\tilde{R}(x)|\leq 2^{O(d)}
\end{equation}
by repeating the argument at the beginning of the proof, which concludes the full proof.
\end{proof}

A missing ingredient in the above proof is to show that when a polynomial in bounded on a set of $O(d^2)$ equally spaced points in an interval, then it can be uniformly bounded on that interval. As discussed in Remark~\ref{rmk:rakhmanov}, the number of points needed can be improved to $O(d)$ by using a powerful result of Rakhmanov~\cite{Rakhmanov2007bounds}.

\begin{lemma}\label{lemma:bernsteinmarkov}
Let $P$ be a degree $d$ polynomial which satisfies $|P(x)|\leq c$ for equally spaced points $x\in\{0,\frac{a}{N},\frac{2a}{N},\dots,a\}$ in the interval $[0,a]$, where $\frac{d^2}{N}\leq 1-\varepsilon$. Then $P$ is uniformly bounded as $|P(x)|\leq\frac{c}{\varepsilon},\forall x\in[0,a]$.
\end{lemma}
\begin{proof}
To prove a uniform bound for $P(x)$, we use a Markov's inequality~\cite{lorentz2005approximation} to bound the maximum derivative,
\begin{equation}
    \max_{x\in[0,a]}|P'(x)|\leq \frac{2d^2}{a}\max_{x\in[0,a]}|P(x)|.
\end{equation}
Let $M=\max_{x\in[0,a]}|P(x)|$ which is achieved at $t\frac{a}{N}\leq x\leq (t+1)\frac{a}{N}$. Then
\begin{nalign}
    \frac{2d^2}{a}M&\geq \max_{x\in[0,a]}|P'(x)|\geq (M-c)\frac{2N}{a},
\end{nalign}
which gives $M\leq \frac{c}{\varepsilon}$.
\end{proof}

We are now ready to fill in the final missing part in the proof of Theorem~\ref{Thm:robustBW} which is in Eq.~\eqref{eq:secondmissing}. A first idea is to directly apply Paturi's Lemma~\ref{Lemma:Paturi} which gives the final error bound $\delta e^{O(d\Delta^{-1})}$. However, this is clearly over pessimistic: for degree $d$ polynomials, the error growth should not be much bigger than $(1/\Delta)^d$ when $\Delta$ is very small. Based on this intuition, our following result gives an improvement for the error bound of long distance polynomial extrapolation.

\begin{lemma}\label{lemma:longdistanceextrapolation}
Let $P$ be a degree $d$ polynomial that satisfies $|P(x)|\leq \delta$, $\forall x\in[0,\Delta]$ ($\Delta<1$). Then we have
\begin{equation}
    |P(1)|\leq \delta e^{d\log\Delta^{-1} + d\log 8}.
\end{equation}
\end{lemma}
\begin{proof}
Let $Q(x)=P\left(\frac{x+1}{2}\Delta\right)$. Then $Q(x)$ is a degree $d$ polynomial that satisfies $|Q(x)|\leq\delta$, $\forall x\in[-1,1]$. Next we use a well-known fact about polynomials, namely that the coefficients of bounded polynomials are at most exponential in the degree. Let $Q(x)=\sum_{i=0}^d a_i x^i$. Lemma \ref{lemma:boundedcoefficient} implies that
\begin{equation}
    \sum_{i=0}^d |a_i|\leq 4^d \delta.
\end{equation}
Then we have
\begin{nalign}
    |P(1)|&=\left|Q\left(\frac{2}{\Delta}-1\right)\right|\\
    &\leq \sum_{i=0}^d |a_i|\left(\frac{2}{\Delta}-1\right)^i\\
    &\leq \sum_{i=0}^d |a_i|\left(\frac{2}{\Delta}\right)^d\\
    &\leq \delta\frac{8^d}{\Delta^d}=\delta e^{d\log\Delta^{-1}+d\log 8}.
\end{nalign}

\end{proof}

\begin{lemma}[Lemma 4.1 in \cite{sherstov2013making}]\label{lemma:boundedcoefficient}
Let $P(x)=\sum_{i=0}^d a_i x^i$ be a polynomial. Then
\begin{equation}
    \sum_{i=0}^d |a_i|\leq 4^d \max_{x\in[-1,1]}|P(x)|.
\end{equation}
\end{lemma}

\subsection{Alternative proof via variable rescaling}\label{sec:alternativeproof}
In Section~\ref{sec:rescaling} we first introduced a proof technique via variable rescaling, which gives the same asymptotic bounds as the results presented above, but yields worse constants. In the following we explicitly calculate the constants when applying this variable rescaling technique to BosonSampling, where the goal is to establish robust average case hardness for random permanents.

We consider the following general framework for polynomial interpolation. Consider a degree $d$ polynomial $P(x)$ that is bounded on a large fraction of points in the interval $[0,\Delta]$ ($\Delta<1$). That is, suppose there is a set of data points $D=\{x_i\}$ which is equally spaced in the interval $[0,\Delta]$, such that
\begin{equation}
    |P(x_i)|\leq\delta
\end{equation}
for at least $1-\eta$ fraction of points in $D$. Here we assume that $|D|$ is a large enough polynomial in $d$ (for example, $|D|=1000d^2$), and $\eta$ is a small enough constant (for example, $\eta=0.001$). Our results in this subsection suggest that $|P(1)|\leq\delta e^{O(d\Delta^{-1})}$, and we will also calculate the constant in the $O$ notation.

The first step is to bound the polynomial $P(x)$ in the interval $[0,\Delta]$. Consider the set of points $\{0,t,2t,\dots,(|D|-1)t\}$ where $t=\frac{\Delta}{|D|}$, and we know that $P(x)$ is bounded at $1-\eta$ fraction of these points. From the proof of Lemma~\ref{Lemma:zeph}, it was shown that without loss of generality, we can assume that the polynomial is bounded at the first $\frac{1}{2}-\eta$ fraction of points in $D$. Combining Lemma~\ref{lemma:bernsteinmarkov} and Lemma~\ref{Lemma:Paturi} which introduces a constant 4 in the exponent, we have
\begin{equation}
    |P(x)|\leq 1.01\delta e^{4d(1/2-\eta)^{-1}} \leq \delta e^{4\times d\times 2.01},\,\,\,\,\forall x\in[0,\Delta].
\end{equation}
This allows us to directly invoke Lemma~\ref{Lemma:Paturi} and obtain
\begin{equation}
    |P(1)|\leq \delta e^{4\times d\times 2.01}\cdot e^{4d\Delta^{-1}}=\delta e^{4d\left(\Delta^{-1}+2.01\right)}.
\end{equation}

Next we plug in the parameters for BosonSampling. The natural polynomial $\Per(X(\theta))$ is a degree $n$ polynomial in $\theta$. Therefore, the polynomial $|\Per(X(\theta^k))|^2$ is a degree $2kn$ polynomial in $\theta$. When $\theta\in[0,\Delta]$, the total variation distance between $X(\theta^k)$ and Gaussian permanents is bounded by $2.01n^2\Delta^k$~\cite{Aaronson2011linear}, which we require to be a small constant (say 0.001). This then gives $\Delta^{-1}=c^{1/k}\cdot n^{2/k}$, for some large constant $c\sim 2000$, and the exponent is
\begin{equation}
    4\cdot 2kn\cdot (2000^{1/k}\cdot n^{2/k}+2.01).
\end{equation}
Let $k=t\log n$ for some constant $t$, then the exponent becomes
\begin{equation}
    8t\cdot (2000^{1/t\log n}\cdot e^{2/t}+2.01)\cdot n\log n.
\end{equation}
Consider large $n$ limit, and optimize with respect to $t$, the constant prefactor is
\begin{equation}
    \min_{t>0} 8t\cdot (1.01\cdot e^{2/t}+2.01)=69.7<70.
\end{equation}

Consider BosonSampling experiment with $n$ photons and $m=n^c$ ($c>2$) output modes. The goal is to prove average case hardness for the output probability up to additive imprecision
\begin{equation}
    \frac{n!}{m^n}\approx e^{-(c-1)n\log n}.
\end{equation}
In the worst case, the output probability is hard within $1/m^n=e^{-c n\log n}$ additive error. Therefore our proof shows average case hardness with robustness
\begin{equation}
    \delta=e^{-(c+70) n\log n}.
\end{equation}

\section*{Acknowledgements}

We especially thank Dorit Aharonov as well as Juani Bermejo-Vega, Jonas Haferkamp, Dominik Hangleiter, and Jens Eisert, for early discussions regarding hardness in the presence of noise that were influential in formulating the questions posed in this manuscript. We also thank Scott Aaronson, Roozbeh Bassirianjahromi, Abhinav Deshpande, Raul Garcia-Patron, Sam Gunn, Ramis Movassagh, Chinmay Nirkhe, Bryan O'Gorman and Umesh Vazirani for helpful comments and discussions. This work was supported by Vannevar Bush faculty fellowship N00014-17-1-3025, DOE QSA grant \#FP00010905 and NSF QLCI Grant No.~2016245. B.F. acknowledges support from AFOSR (YIP number FA9550-18-1-0148 and FA9550-21-1-0008) and the National Science Foundation under Grant CCF-2044923 (CAREER). We would like to thank the Simons Institute for the Theory of Computing where part of this work was done.

\bibliographystyle{alpha}
\bibliography{ref}

\newcommand{\etalchar}[1]{$^{#1}$}
\begin{thebibliography}{QBQGP20}

\bibitem[AA11]{Aaronson2011linear}
Scott Aaronson and Alex Arkhipov.
\newblock The computational complexity of linear optics.
\newblock In {\em Proceedings of the Forty-Third Annual ACM Symposium on Theory
  of Computing}, STOC ’11, page 333–342, New York, NY, USA, 2011.
  Association for Computing Machinery.

\bibitem[AA14]{Aaronson2014Bosonsampling}
Scott Aaronson and Alex Arkhipov.
\newblock Bosonsampling is far from uniform.
\newblock {\em Quantum Info. Comput.}, 14(15–16):1383–1423, November 2014.

\bibitem[AAB{\etalchar{+}}19]{arute2019quantum}
Frank Arute, Kunal Arya, Ryan Babbush, Dave Bacon, Joseph~C Bardin, Rami
  Barends, Rupak Biswas, Sergio Boixo, Fernando~GSL Brandao, David~A Buell,
  et~al.
\newblock Quantum supremacy using a programmable superconducting processor.
\newblock {\em Nature}, 574(7779):505--510, 2019.

\bibitem[Aar05]{Aaronson2005quantum}
Scott Aaronson.
\newblock Quantum computing, postselection, and probabilistic polynomial-time.
\newblock {\em Proceedings of the Royal Society A: Mathematical, Physical and
  Engineering Sciences}, 461(2063):3473--3482, 2005.

\bibitem[ABO08]{aharonov2008fault}
Dorit Aharonov and Michael Ben-Or.
\newblock Fault-tolerant quantum computation with constant error rate.
\newblock {\em SIAM Journal on Computing}, 38(4):1207--1282, 2008.

\bibitem[ABOIN96]{aharonov1996limitations}
Dorit Aharonov, Michael Ben-Or, Russell Impagliazzo, and Noam Nisan.
\newblock Limitations of noisy reversible computation.
\newblock {\em arXiv:quant-ph/9611028}, 1996.

\bibitem[AC17]{aaronson2017complexity}
Scott Aaronson and Lijie Chen.
\newblock Complexity-theoretic foundations of quantum supremacy experiments.
\newblock In {\em Proceedings of the 32nd Computational Complexity Conference},
  CCC ’17, Dagstuhl, DEU, 2017. Schloss Dagstuhl--Leibniz-Zentrum fuer
  Informatik.

\bibitem[AG19]{aaronson2019classical}
Scott Aaronson and Sam Gunn.
\newblock On the classical hardness of spoofing linear cross-entropy
  benchmarking.
\newblock {\em arXiv:1910.12085}, 2019.

\bibitem[AGP07]{aliferis2007accuracy}
Panos Aliferis, Daniel Gottesman, and John Preskill.
\newblock Accuracy threshold for postselected quantum computation.
\newblock {\em arXiv:quant-ph/0703264}, 2007.

\bibitem[AH14]{aaronson2012generalizing}
Scott Aaronson and Travis Hance.
\newblock Generalizing and derandomizing {G}urvits's approximation algorithm
  for the permanent.
\newblock {\em Quantum Information \& Computation}, 14(7\&8):541--559, 2014.

\bibitem[BCG20]{barak2020spoofing}
Boaz Barak, Chi-Ning Chou, and Xun Gao.
\newblock Spoofing linear cross-entropy benchmarking in shallow quantum
  circuits.
\newblock {\em arXiv:2005.02421}, 2020.

\bibitem[BFNV19]{bouland2019complexity}
Adam Bouland, Bill Fefferman, Chinmay Nirkhe, and Umesh Vazirani.
\newblock On the complexity and verification of quantum random circuit
  sampling.
\newblock {\em Nature Physics}, 15(2):159--163, 2019.

\bibitem[BHH16]{brandao2016local}
Fernando~GSL Brandao, Aram~W Harrow, and Micha{\l} Horodecki.
\newblock Local random quantum circuits are approximate polynomial-designs.
\newblock {\em Communications in Mathematical Physics}, 346(2):397--434, 2016.

\bibitem[BIS{\etalchar{+}}18]{boixo2018characterizing}
Sergio Boixo, Sergei~V Isakov, Vadim~N Smelyanskiy, Ryan Babbush, Nan Ding,
  Zhang Jiang, Michael~J Bremner, John~M Martinis, and Hartmut Neven.
\newblock Characterizing quantum supremacy in near-term devices.
\newblock {\em Nature Physics}, 14(6):595--600, 2018.

\bibitem[BJS11]{Bremner2011classical}
Michael~J. Bremner, Richard Jozsa, and Dan~J. Shepherd.
\newblock Classical simulation of commuting quantum computations implies
  collapse of the polynomial hierarchy.
\newblock {\em Proceedings of the Royal Society A: Mathematical, Physical and
  Engineering Sciences}, 467(2126):459--472, 2011.

\bibitem[BMS16]{Bremner2016averagecase}
Michael~J. Bremner, Ashley Montanaro, and Dan~J. Shepherd.
\newblock Average-case complexity versus approximate simulation of commuting
  quantum computations.
\newblock {\em Phys. Rev. Lett.}, 117:080501, Aug 2016.

\bibitem[BMS17]{bremner2017achieving}
Michael~J Bremner, Ashley Montanaro, and Dan~J Shepherd.
\newblock Achieving quantum supremacy with sparse and noisy commuting quantum
  computations.
\newblock {\em Quantum}, 1:8, 2017.

\bibitem[CPS99]{cai1999hardness}
Jin-Yi Cai, Aduri Pavan, and D~Sivakumar.
\newblock On the hardness of permanent.
\newblock In {\em Annual Symposium on Theoretical Aspects of Computer Science},
  pages 90--99. Springer, 1999.

\bibitem[DHJB20]{dalzell2020random}
Alexander~M Dalzell, Nicholas Hunter-Jones, and Fernando~GSL Brand{\~a}o.
\newblock Random quantum circuits anti-concentrate in log depth.
\newblock {\em arXiv:2011.12277}, 2020.

\bibitem[FGHP99]{Fenner1999determining}
Stephen Fenner, Frederic Green, Steven Homer, and Randall Pruim.
\newblock Determining acceptance possibility for a quantum computation is hard
  for the polynomial hierarchy.
\newblock {\em Proceedings of the Royal Society of London. Series A:
  Mathematical, Physical and Engineering Sciences}, 455(1991):3953--3966, 1999.

\bibitem[FR99]{FORTNOW1999complexity}
Lance Fortnow and John Rogers.
\newblock Complexity limitations on quantum computation.
\newblock {\em Journal of Computer and System Sciences}, 59(2):240 -- 252,
  1999.

\bibitem[Fuj16]{fujii2016noise}
Keisuke Fujii.
\newblock Noise threshold of quantum supremacy.
\newblock {\em arXiv:1610.03632}, 2016.

\bibitem[GD18]{gao2018efficient}
Xun Gao and Luming Duan.
\newblock Efficient classical simulation of noisy quantum computation.
\newblock {\em arXiv:1810.03176}, 2018.

\bibitem[GLR{\etalchar{+}}91]{Gemmell1991self}
Peter Gemmell, Richard Lipton, Ronitt Rubinfeld, Madhu Sudan, and Avi
  Wigderson.
\newblock Self-testing/correcting for polynomials and for approximate
  functions.
\newblock In {\em Proceedings of the Twenty-Third Annual ACM Symposium on
  Theory of Computing}, STOC '91, page 33–42, New York, NY, USA, 1991.
  Association for Computing Machinery.

\bibitem[GPRS19]{garcia2019simulating}
Ra{\'u}l Garc{\'\i}a-Patr{\'o}n, Jelmer~J Renema, and Valery Shchesnovich.
\newblock Simulating boson sampling in lossy architectures.
\newblock {\em Quantum}, 3:169, 2019.

\bibitem[GS92]{gemmell1992highly}
Peter Gemmell and Madhu Sudan.
\newblock Highly resilient correctors for polynomials.
\newblock {\em Information processing letters}, 43(4):169--174, 1992.

\bibitem[Gur05]{Gurvits2005on}
Leonid Gurvits.
\newblock On the complexity of mixed discriminants and related problems.
\newblock In {\em Proceedings of the 30th International Conference on
  Mathematical Foundations of Computer Science}, MFCS'05, page 447–458,
  Berlin, Heidelberg, 2005. Springer-Verlag.

\bibitem[Gur06]{guruswami2006list}
Venkatesan Guruswami.
\newblock List decoding in average-case complexity and pseudorandomness.
\newblock In {\em 2006 IEEE Information Theory Workshop-ITW'06 Punta del Este},
  pages 32--36. IEEE, 2006.

\bibitem[HL09]{Harrow2009random}
Aram~W. Harrow and Richard~A. Low.
\newblock Random quantum circuits are approximate 2-designs.
\newblock {\em Communications in Mathematical Physics}, 291(1):257--302, Oct
  2009.

\bibitem[HM18]{harrow2018approximate}
Aram Harrow and Saeed Mehraban.
\newblock Approximate unitary $ t $-designs by short random quantum circuits
  using nearest-neighbor and long-range gates.
\newblock {\em arXiv:1809.06957}, 2018.

\bibitem[HP06]{Hiai2006semicircle}
Fumio Hiai and Denes Petz.
\newblock {\em The Semicircle Law, Free Random Variables and Entropy
  (Mathematical Surveys \& Monographs)}.
\newblock American Mathematical Society, USA, 2006.

\bibitem[Kit97]{Kitaev_1997}
A~Yu Kitaev.
\newblock Quantum computations: algorithms and error correction.
\newblock {\em Russian Mathematical Surveys}, 52(6):1191--1249, dec 1997.

\bibitem[KMM21]{kondo21}
Yasuhiro Kondo, Ryuhei Mori, and Ramis Movassagh.
\newblock Fine-grained analysis and improved robustness of quantum supremacy
  for random circuit sampling.
\newblock {\em arXiv:2102.01960 [quant-ph]}, 2021.

\bibitem[Lip89]{Lipton91new}
Richard~J. Lipton.
\newblock New directions in testing.
\newblock In Joan Feigenbaum and Michael Merritt, editors, {\em Distributed
  Computing And Cryptography, Proceedings of a {DIMACS} Workshop, Princeton,
  New Jersey, USA, October 4-6, 1989}, volume~2 of {\em {DIMACS} Series in
  Discrete Mathematics and Theoretical Computer Science}, pages 191--202.
  {DIMACS/AMS}, 1989.

\bibitem[Lor05]{lorentz2005approximation}
G.G. Lorentz.
\newblock {\em Approximation of Functions}.
\newblock AMS Chelsea Publishing Series. Holt, Rinehart and Winston, 2005.

\bibitem[Mov20]{movassagh2020quantum}
Ramis Movassagh.
\newblock Quantum supremacy and random circuits.
\newblock {\em arXiv:1909.06210v4 [quant-ph]}, 2020.

\bibitem[NPD{\etalchar{+}}20]{napp2020efficient}
John Napp, Rolando L.~La Placa, Alexander~M. Dalzell, Fernando G. S.~L.
  Brandao, and Aram~W. Harrow.
\newblock Efficient classical simulation of random shallow 2{D} quantum
  circuits.
\newblock {\em arXiv:2001.00021}, 2020.

\bibitem[ODMZ20]{oszmaniec2020fermion}
Micha{\l} Oszmaniec, Ninnat Dangniam, Mauro~ES Morales, and Zolt{\'a}n
  Zimbor{\'a}s.
\newblock Fermion sampling: a robust quantum computational advantage scheme
  using fermionic linear optics and magic input states.
\newblock {\em arXiv:2012.15825}, 2020.

\bibitem[ONFJ21]{oh2021classical}
Changhun Oh, Kyungjoo Noh, Bill Fefferman, and Liang Jiang.
\newblock Classical simulation of lossy boson sampling using matrix product
  operators.
\newblock {\em arXiv:2101.11234}, 2021.

\bibitem[Pat92]{Paturi1992on}
Ramamohan Paturi.
\newblock On the degree of polynomials that approximate symmetric boolean
  functions (preliminary version).
\newblock In {\em Proceedings of the Twenty-fourth Annual ACM Symposium on
  Theory of Computing}, STOC '92, pages 468--474, New York, NY, USA, 1992. ACM.

\bibitem[QBQGP20]{qi2020regimes}
Haoyu Qi, Daniel~J Brod, Nicol{\'a}s Quesada, and Ra{\'u}l
  Garc{\'\i}a-Patr{\'o}n.
\newblock Regimes of classical simulability for noisy {G}aussian boson
  sampling.
\newblock {\em Physical review letters}, 124(10):100502, 2020.

\bibitem[Rak07]{Rakhmanov2007bounds}
E.~A. Rakhmanov.
\newblock Bounds for polynomials with a unit discrete norm.
\newblock {\em Annals of Mathematics}, 165(1):55--88, 2007.

\bibitem[Rei06]{reichardt2006error}
Ben~W Reichardt.
\newblock Error-detection-based quantum fault tolerance against discrete
  {P}auli noise.
\newblock {\em arXiv:quant-ph/0612004}, 2006.

\bibitem[She13]{sherstov2013making}
Alexander~A. Sherstov.
\newblock Making polynomials robust to noise.
\newblock {\em Theory of Computing}, 9(18):593--615, 2013.

\bibitem[Sto85]{Stockmeyer1985on}
Larry Stockmeyer.
\newblock On approximation algorithms for \#{P}.
\newblock {\em SIAM Journal on Computing}, 14(4):849--861, 1985.

\bibitem[TD02]{terhal2002adaptive}
Barbara~M Terhal and David~P DiVincenzo.
\newblock Adaptive quantum computation, constant depth quantum circuits and
  {A}rthur-{M}erlin games.
\newblock {\em arXiv:quant-ph/0205133}, 2002.

\bibitem[TO92]{Toda1992counting}
Seinosuke Toda and Mitsunori Ogiwara.
\newblock Counting classes are at least as hard as the polynomial-time
  hierarchy.
\newblock {\em SIAM Journal on Computing}, 21(2):316--328, 1992.

\bibitem[WB86]{welch1986error}
Lloyd~R Welch and Elwyn~R Berlekamp.
\newblock Error correction for algebraic block codes, December~30 1986.
\newblock US Patent 4,633,470.

\bibitem[ZSW20]{zhou2020limits}
Yiqing Zhou, E~Miles Stoudenmire, and Xavier Waintal.
\newblock What limits the simulation of quantum computers?
\newblock {\em Physical Review X}, 10(4):041038, 2020.

\bibitem[ZWD{\etalchar{+}}20]{Zhong2020quantum}
Han-Sen Zhong, Hui Wang, Yu-Hao Deng, Ming-Cheng Chen, Li-Chao Peng, Yi-Han
  Luo, Jian Qin, Dian Wu, Xing Ding, Yi~Hu, Peng Hu, Xiao-Yan Yang, Wei-Jun
  Zhang, Hao Li, Yuxuan Li, Xiao Jiang, Lin Gan, Guangwen Yang, Lixing You,
  Zhen Wang, Li~Li, Nai-Le Liu, Chao-Yang Lu, and Jian-Wei Pan.
\newblock Quantum computational advantage using photons.
\newblock {\em Science}, 370(6523):1460--1463, 2020.

\end{thebibliography}
\appendix

\section{Relationship between truncated Taylor series and Cayley path}\label{appendix:truncatedtaylor}

A line of research initiated by Aaronson and Arkhipov~\cite{Aaronson2011linear} aims to establish complexity-theoretic evidence for the classical hardness of quantum supremacy experiments by providing evidence for the so-called ``approximate sampling" conjecture, which in the context of random circuit sampling can be summarized as follows:

\begin{conjecture}[Approximate sampling]\label{conj:approxsampling}
No polynomial-time algorithm can do the following: on input a random circuit $C\sim\mc H_{\mc A}$, obtain a sample from a distribution $D_C'$ that satisfies $D_{\mathrm{TV}}(D_C,D_C')\leq\frac{1}{\poly(n)}$, with probability at least $1-\frac{1}{\poly(n)}$ over the choice of $C$ as well as the internal randomness of the algorithm. Here $D_C$ denotes the (ideal) output distribution of $C$.
\end{conjecture}

One approach to proving the approximate sampling conjecture is by showing that the output probability of random circuits are $\sharpp$-hard to approximate on average. Here ``approximate" corresponds to a $\frac{2^{-n}}{\poly(n)}$ additive imprecision, which follows naturally from Stockmeyer's approximate counting theorem~\cite{Stockmeyer1985on}.

\begin{conjecture}\label{conj:approxcomputing}
It is $\sharpp$-hard to compute $\przero{C}$ up to additive imprecision $\frac{2^{-n}}{\poly(n)}$ on input a random circuit $C\sim\mc H_{\mc A}$, with probability at least $1-\frac{1}{\poly(n)}$ over the choice of $C$ as well as the internal randomness of the algorithm.
\end{conjecture}

\begin{theorem}[\cite{bouland2019complexity}]
Conjecture~\ref{conj:approxcomputing} implies Conjecture~\ref{conj:approxsampling}.
\end{theorem}

Although Conjecture~\ref{conj:approxcomputing} has not been proven so far, necessary conditions of Conjecture~\ref{conj:approxcomputing} were established in the form of average-case hardness of $\przero{C}$ that is robust up to a smaller additive imprecision. An open direction is therefore to improve these results to match the $\frac{2^{-n}}{\poly(n)}$ additive imprecision required by Conjecture~\ref{conj:approxcomputing}.

As introduced in Section~\ref{section:avghardness}, \cite{bouland2019complexity} developed an approach based on truncated Taylor series. Consider the circuit distribution $\mc H_{\mc A,\theta,K}$ from Definition~\ref{def:thetaperturbeddistribution} which is instantiated by Eq.~\eqref{eq:taylortruncation}. They proved the exact average-case hardness for circuits drawn from $\mc H_{\mc A,\theta,K}$.

\begin{theorem}[\cite{bouland2019complexity}]\label{thm:bfnvexact}
It is $\sharpp$-hard to exactly compute $\frac{3}{4}+\frac{1}{\poly(n)}$ of the probabilities $\przero{C'}$ over the choice of $C'$ from the distribution $\mc H_{\mc A,\theta,K}$ where $\theta=\frac{1}{\poly(n)}$, $K=\poly(n)$.
\end{theorem}

They also showed that the average-case hardness still holds up to an additive imprecision of $2^{-\text{poly}(m)}$, but did not quantify the degree of the polynomial involved\footnote{In retrospect, setting $\theta=\frac{1}{m^2\text{polylog}(m)}$ and $K=\text{polylog}(m)$ yields robustness $2^{-\tilde{O}(m^3)}$ to computing this quantity, but these parameter settings or the arguments justifying them were not given in the original paper.}. Note that here circuits drawn from the distribution $\mc H_{\mc A,\theta,K}$ are not unitary and has a small truncation error from the truncation of the Taylor series in Eq.~\eqref{eq:taylortruncation}. Therefore, Theorem~\ref{thm:bfnvexact} (as well as the $2^{-O(m^3)}$ robust version) does not imply the exact version of Conjecture~\ref{conj:approxcomputing} (that is, the exact average-case hardness of circuits drawn from $\mc H_{\mc A}$). However, as the truncation error is much smaller than the additive imprecision requirement of Conjecture~\ref{conj:approxcomputing}, it is shown that the \emph{approximate} version of Theorem~\ref{thm:bfnvexact} is actually equivalent to Conjecture~\ref{conj:approxcomputing}.

\begin{theorem}[\cite{bouland2019complexity}]\label{thm:bfnvequivalent}
Conjecture~\ref{conj:approxcomputing} is equivalent to the following: it is $\sharpp$-hard to compute $\przero{C'}$ up to additive imprecision $\frac{2^{-n}}{\poly(n)}$ on input a random circuit $C'\sim\mc H_{\mc A,\theta,K}$ where $\theta=\frac{1}{\poly(n)}$, $K=\poly(n)$, with probability at least $1-\frac{1}{\poly(n)}$ over the choice of $C'$ as well as the internal randomness of the algorithm.
\end{theorem}

As discussed in Section~\ref{section:avghardness}, Movassagh's subsequent approach based on the Cayley transform~\cite{movassagh2020quantum} does not have the truncation problem, and therefore they were able to prove the exact version of Conjecture~\ref{conj:approxcomputing}. Movassagh also showed a similar robust version of his theorem with quantified robustness of $2^{-O(m^3)}$ additive imprecision. 
Therefore these two approaches are equivalent from the viewpoint of providing evidence for the approximate sampling conjecture: to prove Conjecture~\ref{conj:approxcomputing}, it is necessary and sufficient to improve the robustness of either one of these approaches to $\frac{2^{-n}}{\poly(n)}$. 

Our first result can be viewed as making progress in this direction. Note that in the main text our proof is based on a generalization of the Cayley path approach, but our techniques can also be applied to truncated Taylor series and obtain similar results. On the other hand, for our second result on noisy circuits, the Cayley path approach is more preferable, as in this case we are no longer in the regime of providing evidence for the approximate sampling conjecture, but are interested in tiny deviations from uniformity in the output distribution of noisy circuits. 

\section{Limitations of polynomial interpolation techniques}
\label{sec:barrier}

Here we give a self-contained description of the barrier to proving hardness of approximate sampling using polynomial interpolation techniques. This argument was first developed by Aaronson and Arkhipov for BosonSampling~\cite[Section 9.2]{Aaronson2011linear}, which can also be naturally generalized to random circuit sampling. They show that, assuming anti-concentration (which is proven for random circuit sampling), the average-case hardness proof based on polynomial interpolation cannot tolerate the amount of error required for showing hardness of approximate sampling.

Consider the output probability of a $n$-photon $m=n^c$ mode BosonSampling experiment. Following the notation in the proof of Corollary~\ref{cor:permanent}, the goal is to prove average-case hardness of computing $p_X=\left|\Per(X)\right|^2/m^n$ for $X\sim\mc G^{n\times n}$. To do so, we consider
\begin{equation}
    X(\theta):=(1-\theta)X_1+\theta X_0
\end{equation}
where $X_0$ is an arbitrary 0/1 matrix and $X_1\sim\mc G^{n\times n}$. The goal is to recover $\left|\Per(X_0)\right|^2$ using polynomial interpolation, given approximate values of $\left|\Per(X(\theta))\right|^2$ for small values of $\theta$. The main idea of the barrier result is to show the following:
\begin{quote}
\emph{There exists a polynomial that agrees well with $\left|\Per(X(\theta))\right|^2$ when $\theta$ is small, but is far from $\left|\Per(X_0)\right|^2$ when $\theta=1$.}
\end{quote}
This statement implies that it is impossible to recover the worst-case value from noisy average-case observations. To construct such a polynomial $P(\theta)$, consider
\begin{nalign}
    &P(\theta):=\left|\Per(X(\theta))\right|^2+t Q(\theta),\\
    &Q(\theta):=\left|\Per((1-\theta)X_1+\theta \mathbf{1}_{n\times n})\right|^2
\end{nalign}
where $\mathbf{1}$ denotes the all-one matrix and $t>0$. Then the error of the polynomial $P(\theta)$ is given by 
\begin{equation}
    \left|P(\theta)-\left|\Per(X(\theta))\right|^2\right|=t\left|Q(\theta)\right|.
\end{equation}
Using the anti-concentration property of $\Per(X_1)$, with high probability the error is given by
\begin{itemize}
    \item average-case error ($\theta\ll 1$): $t\left|Q(\theta)\right|\leq t\cdot\poly(n)\cdot n!$,
    \item worst-case error ($\theta=1$): $t\left|Q(1)\right|=t\left(n!\right)^2$.
\end{itemize}
To prove hardness of approximate sampling, we require $t=1/q(n)$ for a sufficiently large polynomial $q$. However, in this case the worst-case error becomes too large and can be achieved by existing classical algorithms \cite{Gurvits2005on,aaronson2012generalizing}.

On the other hand, choosing $t=1/\left(n!\right)^2$ is sufficient for worst-case hardness. In this case we can tolerate an additive error $\poly(n)/n!$ for average-case permanents. Consider the output probability and normalize the errors by $m^n$, we conclude that the barrier result of Aasonron and Arkhipov does not rule out the possibility of showing hardness of computing the output probability $p_X$ up to additive error
\begin{equation}
    e^{-(c+1)n\log n}.
\end{equation}

The same argument can also be applied to random circuit sampling, which shows that proving hardness of computing $\przero{C}$ up to additive error $2^{-n}/\poly(n)$, which is required for proving hardness of sampling, is impossible using polynomial interpolation techniques. On the other hand, this barrier result does not rule out the possibility of showing hardness of computing $\przero{C}$ up to additive error $2^{-O(n)}$. 

Interestingly, using essentially the same techniques, our result is a log factor far from the barrier for random circuit sampling, while for BosonSampling our result is only a constant factor away.

\section{Worst-case hardness for noisy circuits from error detection}\label{appendix:worstcasehardness}
It is a well-known result that computing the output probability $\przero{C}$ for quantum circuits $C$ is $\sharpp$-hard in the worst case~\cite{FORTNOW1999complexity,terhal2002adaptive,Bremner2011classical,Bremner2016averagecase}, where the hardness is robust under constant multiplicative approximation or exponentially small additive approximation. Here, we consider a slightly different approach based on \cite{Fenner1999determining}. There, it was shown that determining if $\przero{C}\neq 0$ is hard for the complexity class $\cocequalp$, which corresponds to deciding whether an efficiently computable Boolean function $f:\{0,1\}^n\to\{-1,1\}$ satisfies
\begin{equation}
    \sum_{x\in\{0,1\}^n} f(x)\neq 0.
\end{equation}
Importantly, it was shown~\cite{Toda1992counting} that $\cocequalp$ is hard for $\PH$ under randomized reductions. Therefore, determining if $\przero{C}\neq 0$ is hard for $\PH$ unless $\PH$ collapses.

In the following, we show that we can assume an exponentially small gap in the above decision problem without loss of generality. Consider the following promise decision problem.
\begin{problem}\label{prob:decision}
    On input a $n$-qubit quantum circuit $C$, decide if $\przero{C}=0$, or $\przero{C}\geq\frac{1}{2^{2n}}$.
\end{problem}
We show that this gapped variant is still hard for $\PH$ following a similar argument as in~\cite{Fenner1999determining}.

\begin{lemma}
Problem~\ref{prob:decision} is hard for $\cocequalp$ under a polynomial time reduction.
\end{lemma}
\begin{proof}
Consider an efficiently computable Boolean function $f:\{0,1\}^{n_0}\to\{-1,1\}$. Using the standard technique for classical reversible computation, we can create a $n$-qubit ($n\geq n_0$) unitary circuit $U_f$ in polynomial time, such that $U_f\ket{x}\ket{0}=f(x)\ket{x}\ket{0}$ for all $x\in\{0,1\}^{n_0}$. Let $C$ be the Fourier sampling circuit, $C=H^{\otimes n_0}U_f H^{\otimes n_0}$. Then we have
\begin{equation}
    \przero{C}=\frac{1}{2^{2n_0}}\left(\sum_{x\in\{0,1\}^{n_0}}f(x)\right)^2.
\end{equation}
As $\sum_{x\in\{0,1\}^{n_0}}f(x)$ takes integer value, we know that either $\przero{C}=0$, or $\przero{C}\geq\frac{1}{2^{2n_0}}\geq\frac{1}{2^{2n}}$. Therefore, deciding Problem~\ref{prob:decision} on input $C$ is equivalent to deciding if $\sum_{x\in\{0,1\}^{n_0}}f(x)\neq 0$.
\end{proof}

Next, we show that simulating a noisy circuit up to additive error $\exp\left(-O(m\log m)\right)$ is as hard as Problem~\ref{prob:decision}, and therefore is hard for $\PH$. The proof requires the use of fault-tolerant quantum error detection. Given an arbitrary noise model $\mc N$ in the form of Eq.~\eqref{eq:noisemodel}, its local and stochastic property implies that there exists a constant $\gamma_{\mathrm{th}}>0$ such that fault-tolerant quantum error detection is possible when $\gamma<\gamma_{\mathrm{th}}$. Different from quantum error correction, the fault-tolerance by error detection is \emph{conditional}, or via \emph{post selection}. Given a logical circuit $C_0$ of $n_0$ qubits and $m_0$ gates, suppose $C$ implements $C_0$ with an error detection code using $n$ qubits and $m$ gates. Suppose that $C$ has two output registers $x,z$, where $x$ denotes the logical output, and $z$ denotes the error detection syndrome, where $z=0$ implies that no error is detected throughout the circuit. The fault-tolerance property guarantees that conditioned on no error being detected, the output distribution of $C$ is close to the logical circuit $C_0$, that is,
\begin{equation}
    \Pr_{C_0}[x]\approx\Pr_{C,\mc N}[x|z=0]=\frac{\Pr_{C,\mc N}[x,z=0]}{\Pr_{C,\mc N}[z=0]},\,\,\,\,x\in\{0,1\}^{n_0}.
\end{equation}
Here $\Pr_{C_0}$ denotes the output distribution of the logical circuit $C_0$ and $\Pr_{C,\mc N}$ denotes the output distribution of $C$ with noise model $\mc N$.

As we eventually want to decide Problem~\ref{prob:decision} which has an exponentially small gap, we would like $C$ to be a very accurate simulation of $C_0$. This can be done by choosing a code distance $d=\poly(n_0)$: while introducing a polynomial overhead $n=\poly(n_0)$ and $m=\poly(m_0)$, this guarantees that the total variation distance between the output distribution of $C_0$ and the output distribution of $C$ (conditioned on $z=0$) is bounded by $\exp(-\poly(m_0))$. We thus have the following property, which follows from standard threshold theorems~\cite{Kitaev_1997,reichardt2006error,aliferis2007accuracy,aharonov2008fault}.

\begin{lemma}\label{lemma:errordetection}
Fix a noise model $\mc N$, let $\gamma_{\mathrm{th}}$ be the corresponding fault-tolerant error detection threshold, and suppose the noise rate of $\mc N$ satisfies $\gamma<\gamma_{\mathrm{th}}$. For any logical circuit $C_0$ of $n_0$ qubits and $m_0$ gates and a desired polynomial $p$, there exists a circuit $C$ that can be constructed in polynomial time, with $n=\poly(n_0)$ qubits and $m=\poly(m_0)$ gates, such that
\begin{equation}
    \left|\przero{C_0}-\frac{\przero{C,\mc N}}{\Pr_{C,\mc N}[z=0]}\right|\leq\exp(-p(m_0)).
\end{equation}
\end{lemma}

Now we are ready to prove the worst-case hardness result.

\begin{theorem}\label{thm:worstcasehardness}
Fix a noise model $\mc N$, let $\gamma_{\mathrm{th}}$ be the corresponding fault-tolerant error detection threshold, and suppose the noise rate of $\mc N$ satisfies $\gamma<\gamma_{\mathrm{th}}$. It is $\cocequalp$-hard to compute $\przero{C,\mc N}$ on input $C$ with $m$ gates, up to additive error $\exp(-O(m\log m))$.
\end{theorem}
\begin{proof}
We prove a polynomial time reduction from Problem~\ref{prob:decision}. On input a logical circuit $C_0$ of $n_0$ qubits and $m_0$ gates, we would like to decide if $\przero{C_0}=0$, or $\przero{C_0}\geq\frac{1}{2^{2n_0}}$. Choose $p$ such that $\exp(-p(m_0))\ll\frac{1}{2^{2n_0}}$, and construct a circuit $C$ with $n=\poly(n_0)$ qubits and $m=\poly(m_0)$ gates according to Lemma~\ref{lemma:errordetection}. Consider the probability $\Pr_{C,\mc N}[z=0]$ of detecting no errors. We bound it in two directions. First,
\begin{equation}
    \Pr_{C,\mc N}[z=0]\geq\Pr[\text{no error occurs}]=\exp(-O(\gamma m)).
\end{equation}
On the other hand,
\begin{nalign}
    \Pr_{C,\mc N}[z=0]&=\Pr[\text{no error occurs}]+\Pr[\text{some error occurs but undetected}]\\
    &\leq \exp(-O(\gamma m))\left(1+\exp(-\poly(m_0))\right),
\end{nalign}
where the second term takes into account both the probability of the error pattern as well as the probability of beating the code distance (also see~\cite{fujii2016noise}). Next, consider the two cases of Problem~\ref{prob:decision},
\begin{enumerate}
    \item When $\przero{C_0}=0$, we have
    \begin{nalign}
        \przero{C,\mc N}&\leq \exp(-p(m_0))\Pr_{C,\mc N}[z=0]\\
        &\leq \exp(-p(m_0))\exp(-O(\gamma m)).
    \end{nalign}
    \item When $\przero{C_0}\geq\frac{1}{2^{2n_0}}$, we have
    \begin{nalign}
        \przero{C,\mc N}&\geq\left(\frac{1}{2^{2n_0}} -\exp(-p(m_0))\right)\Pr_{C,\mc N}[z=0]\\
        &\geq\left(\frac{1}{2^{2n_0}} -\exp(-p(m_0))\right) \exp(-O(\gamma m)).
    \end{nalign}
\end{enumerate}
The above bounds can be explicitly computed. Moreover, they have a large gap, which is at least
\begin{equation}
    O\left(\frac{1}{2^{2n_0}}\right) \exp(-O(\gamma m))\gg \exp(-O(m\log m)).
\end{equation}
Therefore, computing $\przero{C,\mc N}$ up to additive error $\exp(-O(m\log m))$ is sufficient for deciding Problem~\ref{prob:decision}.
\end{proof}

Using the fact that $\cocequalp$ is hard for $\PH$ under randomized reductions~\cite{Toda1992counting}, Theorem~\ref{thm:mainworstcasehardness} is a direct corollary of Theorem~\ref{thm:worstcasehardness}.

\section{Proof of Lemma~\ref{lemma:noisyconvergence}}\label{appendix:noisyconvergence}
The proof uses techniques that are developed by Harrow and Low~\cite{Harrow2009random}, which were developed to show that random quantum circuits converge to unitary 2-design. Below we show that noisy circuits correspond to entirely different convergence behavior. These techniques can also be applied to analyze local random circuits, which we leave as important future work.

Consider the quantity $CP$ (collision probability/classical purity) defined as follows,
\begin{equation}
    CP=2^n \E_{U\sim\mathrm{GRC}(n,d)}\left[\sum_x p_{e}^2(x)\right]-1.
\end{equation}
Here $\mathrm{GRC}(n,d)$ denotes the global random circuit distribution on $n$ qubits with depth $d$, as shown in Fig.~\ref{fig:randomcircuitsampling}b, and $p_e$ denotes the output distribution of the noisy circuit. We will drop the subscript $U\sim\mathrm{GRC}(n,d)$ unless necessary. For simplicity, here we consider i.i.d. depolarizing noise, where each blue dot represents a single qubit noise channel
\begin{equation}
    \mc E(\rho)=(1-\gamma)\rho + \gamma \frac{I}{2}\Tr[\rho],
\end{equation}
where $\gamma$ is the noise rate and $I$ is the identity matrix. Using Cauchy–Schwarz inequality, it is easy to see that for each specific circuit, the classical purity is an upper bound of the total variation distance, as
\begin{equation}
    D_{\mathrm{TV}}(p_e,p_{\mathrm{uniform}})=\frac{1}{2}\sum_{x}\left|p_e(x)-\frac{1}{2^n}\right|\leq \frac{1}{2}\sqrt{2^n\sum_x p_e^2(x)-1}.
\end{equation}

To see why we use $CP$ to represent the distance, notice that $CP$ is a linear function of the two-copy output density matrix $\rho_e\otimes \rho_e$. Therefore, by linearity of expectation, we can consider an initial state $\ketbra{0^n}^{\otimes 2}$ going through an average quantum channel; that is, to understand the convergence behavior of $CP$, it suffices to track the evolution of the second moment of the density matrix $\E[\rho_e\otimes \rho_e]$ as the circuit depth increases. More specifically, we can write any Hermitian operator in the Fourier basis as
\begin{equation}
    \E[\rho_e\otimes \rho_e]=\frac{1}{2^n}\sum_{p,q}\delta(p,q)\sigma_p\otimes \sigma_q,
\end{equation}
where $p,q\in\{0,1,2,3\}^{n}$ are $2n$-bit strings, $\sigma_p,\sigma_q\in\{I,X,Y,Z\}^{\otimes n}$ are the corresponding Pauli operators indexed by $0\to I,1\to X,2\to Y,3\to Z$, and
\begin{equation}
    \delta(p,q)=\frac{1}{2^n}\Tr[\E[\rho_e\otimes \rho_e]\sigma_p\otimes \sigma_q]=\frac{1}{2^n}\E\Tr[\left(\rho_e\otimes \rho_e\right)\sigma_p\otimes \sigma_q]
\end{equation}
are the Fourier coefficients, where the second equality follows from linearity.

As we take expectation over random quantum circuits $U\sim\mathrm{GRC}(n,d)$, the Fourier coefficients $\delta(p,q)$ only depend on $(\mc E,n,d)$, i.e. noise channel, system size, and circuit depth, and they contain all the information needed to evaluate $CP$. More specifically we have
\begin{nalign}
    &CP=2^n \E\left[\sum_x p_{e}^2(x)\right]-1=2^n\sum_{p\in\{0,3\}^n\setminus\{0^n\}}\delta(p,p)=2^n\sum_{p\in\{0,3\}^n\setminus\{0^n\}}\alpha(p),
\end{nalign}
where we have defined $\alpha(p):=\delta(p,p)$.

In the following we will study the evolution of these coefficients $\delta(p,q)$ as gates and noise are applied. For the initial state $\ketbra{0^n}^{\otimes 2}$, we have $\delta(p,q)=\frac{1}{2^n}$ for any $p,q\in\{0,3\}^n$ and 0 otherwise. Next, consider the application of a global random unitary on the Pauli basis, it was shown \cite{Harrow2009random} that
\begin{nalign}\label{eq:twoqubitgateaction}
    &\E_{U\sim \mathbb{U}(2^n)}\left[U\sigma_1 U^\dag\otimes U\sigma_2 U^\dag\right],\,\,\,\,\sigma_1,\sigma_2\in\{I,X,Y,Z\}^{\otimes n}\\
    &=\begin{cases}
        0, & \sigma_1\neq \sigma_2,\\
        I^{\otimes 2n}, & \sigma_1= \sigma_2=I^{\otimes n},\\
        \frac{1}{2^{2n}-1}\left(\sum_{p\in\{0,1,2,3\}^n\setminus\{0^n\}}\sigma_p\otimes \sigma_p\right),&\sigma_1= \sigma_2\neq I^{\otimes n}.
    \end{cases}
\end{nalign}
Therefore the off-diagonal coefficients $\delta(p,q)$ ($p\neq q$) becomes 0 for depth $d\geq 1$. From now on we only consider diagonal coefficients $\alpha(p)=\delta(p,p)$. The evolution of $\alpha$ under a global random unitary is therefore given by
\begin{equation}
    \alpha'(p)=\begin{cases}
            \alpha(p), & p=0^n,\\
            \frac{1}{2^{2n}-1}\sum_{q\in\{0,1,2,3\}^n\setminus\{0^n\}}\alpha(q), & p\neq 0^n.
            \end{cases}
\end{equation}
It is easy to observe the following properties: (1) $\alpha(0^n)=\frac{1}{2^n}$ is a constant; (2) for $p\neq 0^n$, $\alpha'(p)$ does not depend on $p$; (3) $\sum_{p\in\{0,1,2,3\}^n}\alpha(p)$ is preserved by unitary evolution.

Finally, consider the effect of depolarizing noise on a single qubit, which gives
\begin{nalign}
    &I\otimes I\to I\otimes I,\\
    &X\otimes X\to (1-\gamma)^2 X\otimes X,\\
    &Y\otimes Y\to (1-\gamma)^2 Y\otimes Y,\\
    &Z\otimes Z\to (1-\gamma)^2 Z\otimes Z.
\end{nalign}
Therefore the effect of noise can be summarized as
\begin{nalign}
    \alpha(p)\to (1-\gamma)^{2|p|}\alpha(p),
\end{nalign}
where $|p|$ denotes the Hamming weight of $p$, i.e. the number of non-zero terms.

Combining everything together, we can write the evolution of $\sum_{p\in\{0,1,2,3\}^n\setminus\{0^n\}}\alpha(p)$ in one cycle (i.e. one global gate followed by a layer of noise channels) as
\begin{equation}
    \sum_{p\in\{0,1,2,3\}^n\setminus\{0^n\}}\alpha(p)\to \frac{\sum_{p\in\{0,1,2,3\}^n\setminus\{0^n\}}(1-\gamma)^{2|p|}}{2^{2n}-1}\sum_{p\in\{0,1,2,3\}^n\setminus\{0^n\}}\alpha(p).
\end{equation}
Define the decay coefficient as
\begin{equation}
    \beta=\frac{\sum_{p\in\{0,1,2,3\}^n\setminus\{0^n\}}(1-\gamma)^{2|p|}}{2^{2n}-1}=\frac{\sum_{k=1}^n \binom{n}{k}3^k(1-\gamma)^{2k}}{2^{2n}-1}=\frac{\left(1+3(1-\gamma)^2\right)^n-1}{2^{2n}-1}.
\end{equation}
Therefore at depth $d$ we have
\begin{equation}
    \sum_{p\in\{0,1,2,3\}^n\setminus\{0^n\}}\alpha(p)=\frac{2^n-1}{2^n}\beta^d.
\end{equation}
To compute $CP$ at depth $d$, note that we can first compute $\sum_{p\in\{0,1,2,3\}^n\setminus\{0^n\}}\alpha(p)$ at depth $d-1$. These coefficients are then uniformly redistributed by the global unitary, and then decay according to noise. Overall we have
\begin{nalign}
    CP&=2^n\sum_{p\in\{0,3\}^n\setminus\{0^n\}}\alpha(p)\\
    &=2^n\sum_{p\in\{0,3\}^n\setminus\{0^n\}}(1-\gamma)^{2|p|}\frac{\frac{2^n-1}{2^n}\beta^{d-1}}{2^{2n}-1}\\
    &=\frac{\left(1+(1-\gamma)^2\right)^n-1}{2^{n}+1}\left(\frac{\left(1+3(1-\gamma)^2\right)^n-1}{2^{2n}-1}\right)^{d-1}.
\end{nalign}
To obtain an upper bound of $CP$, we have
\begin{nalign}
    CP&=\frac{\left(1+(1-\gamma)^2\right)^n-1}{2^{n}+1}\left(\frac{\left(1+3(1-\gamma)^2\right)^n-1}{2^{2n}-1}\right)^{d-1}\\
    &\leq \left(1-\left(\gamma-\frac{\gamma^2}{2}\right)\right)^n\cdot \left(1-\left(\frac{3\gamma}{2}-\frac{3\gamma^2}{4}\right)\right)^{nd-n}\\
    &=\exp\left(-\Omega(\gamma n (d-d_0))\right).
\end{nalign}
where $d_0$ is a universal constant.

\section{Random unitary and Cayley transform}\label{appendix:cayleytransform}
In this section we introduce the Cayley transform~\cite{movassagh2020quantum} and establish some useful properties. These properties are originally proved in~\cite{movassagh2020quantum}, and here we give a self-consistent presentation as we are using different notations.

We consider the Haar measure $\mu_H$ over the unitary group $\mathbb{U}(N)$ of $N\times N$ unitary matrices.
Any unitary matrix can be diagonalized as $U=V\Lambda V^\dag$, where $V$ is unitary and $\Lambda=\mathrm{diag}(e^{i\varphi_1},\dots,e^{i\varphi_N})$, $\varphi_j\in[-\pi,\pi]$. From now on we use ``eigenvalues" to refer to the angles $\varphi_j$. Consider the following transform for the eigenvalues,
\begin{nalign}\label{eq:eigvaltransform}
&f_\theta:[-\pi,\pi]\to[-\pi,\pi],\\
&f_\theta(\varphi)=2\arctan\left((1-\theta)\tan\frac{\varphi}{2}\right).
\end{nalign}
It is easy to verify that applying $f_\theta$ to all eigenvalues of a unitary matrix is equivalent to applying the Cayley transform as given in Definition~\ref{def:cayleytransform}.

Define $\mu_\theta$ ($\theta\geq 0$) as the probability measure on $\mathbb{U}(N)$ induced by applying $f_\theta$ to all eigenvalues of the unitary matrix drawn from $\mu_H$. From the definition it is easy to see that $\mu_0=\mu_H$ and $\mu_1$ is the indicator for the identity matrix. 

In the following we establish the bound of the total variation distance between $\mu_\theta$ and $\mu_H$ when $\theta$ is a small positive number.

\begin{lemma}\label{lemma:totalvariationdistance}
For $N=O(1)$ and $\theta\in(0,0.1)$, the total variation distance between $\mu_\theta$ and $\mu_H$ is bounded by
\begin{equation}
    D_{\mathrm{TV}}(\mu_\theta,\mu_H)=\int_{\mathbb{U}(N)}\frac{1}{2}\left|1-\frac{d\mu_\theta}{d\mu_H}\right|d\mu_H=O\left(\theta\right).
\end{equation}
\end{lemma}
\begin{proof}

\begin{figure}[t]
    \centering
    \includegraphics[width=10cm]{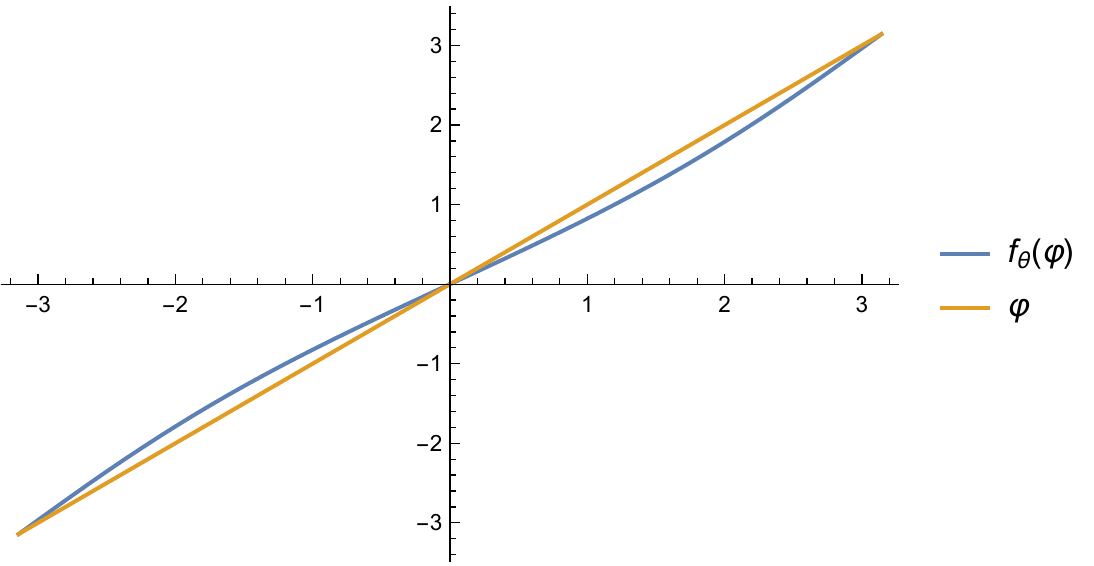}
    \caption{Plot of the eigenvalue transform $f_\theta(\varphi)$ for $\theta=0.2$.}
    \label{fig:transform}
\end{figure}
First, notice that $\mu_\theta$ and $\mu_H$ has the same eigenvector distribution. As the eigenvalues and eigenvectors are independently distributed in $\mu_H$, the integral only concerns the eigenvalue density.

To see this, we use the two-sided invariance of the Haar measure. Let $U\sim_{\mu_H}\mathbb{U}(N)$ be expressed as $U=V\Lambda V^\dag$, where $V$ is the eigenvectors and $\Lambda$ is the eigenvalues. For any unitary $Q\in\mathbb{U}(N)$, $QUQ^\dag$ is equal to $U$ in distribution. Therefore, $QUQ^\dag$ is also equal to $U$ in distribution when $Q$ is independently drawn from $\mu_H$. As $QUQ^\dag=(QV)\Lambda(QV)^\dag$, the eigenvectors $V$ is distributed the same as $QV$ up to global phase, which is distributed the same as $Q$, which is independent from $\Lambda$. To summarize, the eigenvectors and eigenvalues of Haar random unitary are independently distributed, where the eigenvectors are distributed the same as $\mu_H$ up to global phase, and we quote the eigenvalue density from~\cite{Hiai2006semicircle}:
\begin{equation}
g_H(\varphi_1,\dots,\varphi_N)=\frac{1}{(2\pi)^N N!} \prod_{1\leq k<l\leq N}
    \left|e^{i\varphi_k}-e^{i\varphi_l}\right|^2
\end{equation}
where $\varphi_i\in[-\pi,\pi]$. Let $g_\theta$ be the eigenvalue density of $\mu_\theta$. The total variation distance is then reduced to the following integral over Lebesgue measure on $[-\pi,\pi]^N$,
\begin{equation}
    D_{\mathrm{TV}}(\mu_\theta,\mu_H)=\int_{[-\pi,\pi]^N}\frac{1}{2}\left|g_\theta(\varphi_1,\dots,\varphi_N)-g_H(\varphi_1,\dots,\varphi_N)\right|d\varphi_1\cdots d\varphi_N.
\end{equation}

An example of the transform $f_\theta$ is plotted in Fig.~\ref{fig:transform}. It is easy to prove that $f_\theta$ is monotonic and has range $[-\pi,\pi]$.

To prove the upper bound of the total variation distance, we need to calculate the transformed eigenvalue density $g_\theta(\varphi_1,\dots,\varphi_N)$, by applying $f_\theta$ to $\varphi_1,\dots,\varphi_N$ jointly distributed according to $g_H$. Then $g_\theta$ can be written as
\begin{nalign}
    g_\theta(\varphi_1,\dots,\varphi_N)&=g_H\left(f_\theta^{-1}(\varphi_1),\dots,f_\theta^{-1}(\varphi_N)\right)\prod_{i=1}^N\left|\frac{d f_\theta^{-1}(\varphi_i)}{d\varphi_i}\right|,\,\,\,\,\forall \varphi_i\in [-\pi,\pi].
\end{nalign}
where the inverse transform of $f_\theta$ is given by
\begin{equation}
    f_{\theta}^{-1}(\varphi)=2\arctan\left(\frac{1}{1-\theta}\tan\frac{\varphi}{2}\right)
\end{equation}
and its derivative given by
\begin{equation}
    \frac{df^{-1}_\theta(\varphi)}{d\varphi}=\frac{1-\theta}{1-(2\theta-\theta^2)\cos^2\frac{\varphi}{2}}.
\end{equation}
It is easy to prove the following bounds by a simple calculation: assume $\theta$ is small enough (say $\theta<0.1$), then for all $\varphi\in[-\pi,\pi]$,
\begin{nalign}\label{eq:cayleytransformbounds}
    &\left|f_\theta^{-1}(\varphi)-\varphi\right|\leq O(\theta)\\
    &\left|\frac{df^{-1}_\theta(\varphi)}{d\varphi}-1\right|\leq O(\theta),
\end{nalign}
where the constants in the above two equations is independent of $\varphi$. The proof then follows by directly bounding the difference between the densities. Consider
\begin{nalign}
    &\left|g_\theta(\varphi_1,\dots,\varphi_N)-g_H(\varphi_1,\dots,\varphi_N)\right|\\
    &=\left|g_H\left(f_\theta^{-1}(\varphi_1),\dots,f_\theta^{-1}(\varphi_N)\right)\prod_{i=1}^N\left|\frac{d f_1^{-1}(\varphi_i)}{d\varphi_i}\right|-g_H(\varphi_1,\dots,\varphi_N)\right|\\
    &\leq \left|\left(g_H\left(\varphi_1,\dots,\varphi_N\right)+\sum_i O\left(\theta\right)\right)\prod_{i=1}^N\left(1+O\left(\theta\right)\right)-g_H(\varphi_1,\dots,\varphi_N)\right|\\
    &\leq O(\theta),
\end{nalign}
where the third line follows from Eq.~\eqref{eq:cayleytransformbounds} and the continuity of $g_H$. Finally the integral of the difference between the densities is still bounded by $O(\theta)$,
which concludes the full proof.
\end{proof}

Note that in the matrix form of the Cayley transform
\begin{equation}
    H(\theta)=\sum_j\frac{1+i(1-\theta)\tan\frac{\varphi_j}{2}}{1-i(1-\theta)\tan\frac{\varphi_j}{2}}\ketbra{\psi_j},
\end{equation}
where $H=\sum_j e^{i\varphi_j}\ketbra{\psi_j}$, the magnitude of the denominator in the above equation goes to infinity as $\varphi_j$ is close to $\pm\pi$, which will affect the accuracy of the polynomial interpolation in our worst-to-average-case reduction. However, the probability of such an event is small when $H$ is drawn from $\mu_H$, as shown in the following lemma.

\begin{lemma}\label{lemma:haardistributionrange}
The Haar distribution $\mu_H$ over the unitary group $\mathbb{U}(N)$ satisfies
\begin{nalign}
\Pr_{\mu_H}\left[\varphi_j\in [-\pi+\delta,\pi-\delta],\,\,\forall j\right]\geq 1-\frac{N\delta}{\pi},
\end{nalign}
where $\varphi_j$ ($j=1\dots N$) denotes the eigenvalues of the random unitary distributed according to $\mu_H$.
\end{lemma}
\begin{proof}
The proof is a simple union bound:
\begin{nalign}
    \Pr_{\mu_H}\left[\varphi_j\in [-\pi+\delta,\pi-\delta],\,\,\forall j\right]&=1-\Pr_{\mu_H}\left[\exists j:\varphi_j\notin [-\pi+\delta,\pi-\delta]\right]\\
    &\geq 1-\sum_{j=1}^N\Pr_{\mu_H}\left[\varphi_j\notin [-\pi+\delta,\pi-\delta]\right]\\
    &= 1-N\Pr_{\mu_H}\left[\varphi_1\notin [-\pi+\delta,\pi-\delta]\right]\\
    &=1-\frac{N\delta}{\pi}.
\end{nalign}
\end{proof}

\end{document}